\documentclass[10pt,journal,compsoc]{IEEEtran}

\usepackage[normalem]{ulem}
\usepackage{tikz}
\usepackage{url}
\usepackage{pgfplots}
\usepackage{amsmath}
\usepackage{multirow}
\usepackage{threeparttable}
\usepackage{enumitem}
\usepackage{listings}
\usepackage{subfig}
\usepackage{amsthm}

\usepackage[ruled,linesnumbered,lined,boxed,commentsnumbered]{algorithm2e}
\usepgfplotslibrary{statistics}

\usepackage{caption}
\definecolor{mygreen}{rgb}{0,0.6,0}
\newtheorem{theorem}{Theorem}
\newcounter{theorem_counter}
\newtheorem{lemma}{Lemma}
\newcounter{lemma_counter}
\begin{document}

\title{Taskflow: A Lightweight Parallel and Heterogeneous Task Graph Computing System}


\author{Tsung-Wei Huang, Dian-Lun Lin, Chun-Xun Lin, and Yibo Lin
\IEEEcompsocitemizethanks{\IEEEcompsocthanksitem Tsung-Wei Huang, and Dian-Lun Lin are with the Department of Electrical and Computer Engineering, the University of Utah, Salt Lake City, UT.
\IEEEcompsocthanksitem Chun-Xun Lin is with MathWorks, USA.
\IEEEcompsocthanksitem Yibo Lin is with the Department of Computer Science, Peking University, Beijing, China.
}}

\markboth{Transactions on Parallel and Distributed Systems}%
{}

\IEEEtitleabstractindextext{%
\begin{abstract}
Taskflow aims to streamline the building of parallel and heterogeneous
applications using a lightweight task graph-based approach.
Taskflow introduces an expressive task graph programming model
to assist developers in the implementation of parallel and heterogeneous 
decomposition strategies on a heterogeneous computing platform.
Our programming model distinguishes itself as a very general class of task graph 
parallelism with in-graph control flow
to enable end-to-end parallel optimization.
To support our model with high performance,
we design an efficient system runtime that solves many of the new scheduling
challenges arising out of our models
and optimizes the performance across
latency, energy efficiency, and throughput.
We have demonstrated the promising performance of Taskflow 
in real-world applications.
As an example,
Taskflow solves a large-scale machine learning workload up to 
29\% faster,
1.5$\times$ less memory,
and 1.9$\times$ higher throughput
than the industrial system, oneTBB,
on a machine of 40 CPUs and 4 GPUs.
We have opened the source of Taskflow and deployed it to large numbers of
users in the open-source community.

\end{abstract}

}

\maketitle

\section{Introduction}

\IEEEPARstart{T}{ask} graph computing system (TGCS) plays
an essential role in advanced scientific computing.
Unlike loop-based models, TGCSs
encapsulate function calls and their dependencies in a top-down
task graph to implement \textit{irregular} parallel decomposition strategies
that scale to large numbers of processors,
including manycore central processing units (CPUs) and graphics processing units (GPUs).
As a result, recent years have seen a great deal amount of TGCS research,
just name a few,
oneTBB FlowGraph~\cite{TBB}, StarPU~\cite{StarPU},
TPL~\cite{TPL}, Legion~\cite{Legion}, Kokkos-DAG~\cite{Kokkos}, 
PaRSEC~\cite{PaRSEC}, HPX~\cite{HPX}, and Fastflow~\cite{Fastflow}.
These systems have enabled vast success in a variety of 
scientific computing applications,
such as machine learning, data analytics, and simulation.

However, three key limitations
prevent existing TGCSs from exploring the full potential of task graph parallelism.
First, existing TGCSs closely rely on directed acyclic graph (DAG) models
to define tasks and dependencies.
Users implement \textit{control-flow} decisions outside the graph description,
which typically results in rather complicated implementations that lack
\textit{end-to-end parallelism}.
For instance, when encountering an if-else block,
users need to synchronize the graph execution with a TGCS runtime,
which could otherwise be omitted if in-graph control-flow tasks are supported.
Second, 
existing TGCSs do not align well with modern hardware.
In particular, new GPU task graph parallelism, such as CUDA Graph,
can bring significant yet largely untapped performance benefits.
Third,
existing TGCSs are good at either CPU- or GPU-focused workloads, but rarely both
simultaneously.
Consequently, we introduce in this paper \textit{Taskflow},
a lightweight TGCS to overcome these limitations.
We summarize three main contributions of Taskflow as follows:

\begin{itemize}[leftmargin=*]\itemsep=2pt

\item \textbf{Expressive programming model} --
We design an expressive task graph programming model 
by leveraging modern C++ closure.
Our model enables efficient implementations of parallel and heterogeneous 
decomposition strategies using the task graph model.
The expressiveness of our model lets developers perform rather a lot of work
with relative ease of programming.
Our user experiences lead us to believe that, although it requires some effort to learn,
a programmer can master our APIs needed for many applications in just a few hours.

\item \textbf{In-graph control flow} --
We design a new conditional tasking model to support \textit{in-graph control flow}
beyond the capability of traditional DAG models
that prevail in existing TGCSs.
Our condition tasks enable developers to 
integrate control-flow decisions, 
such as conditional dependencies, cyclic execution, and non-deterministic flows
into a task graph of end-to-end parallelism.
In case applications have frequent dynamic behavior,
such as optimization and branch and bound,
programmers can efficiently overlap tasks
both inside and outside the control flow
to hide expensive control-flow costs.

\item \textbf{Heterogeneous work stealing} --
We design an efficient work-stealing algorithm to adapt
the number of workers to dynamically generated task parallelism
at any time during the graph execution.
Our algorithm prevents the graph execution from underutilized
threads that is harmful to performance,
while avoiding excessive waste of thread resources when available tasks are scarce.
The result largely improves the overall system performance,
including latency, energy usage, and throughput.
We have derived theory results to justify the efficiency of our work-stealing
algorithm.
\end{itemize}

We have evaluated Taskflow on real-world applications
to demonstrate its promising performance.
As an example,
Taskflow solved a large-scale machine learning problem up to 
29\% faster,
1.5$\times$ less memory,
and 1.9$\times$ higher throughput
than the industrial system, oneTBB~\cite{TBB},
on a machine of 40 CPUs and 4 GPUs.
We believe Taskflow stands out as a unique system given
the ensemble of software tradeoffs and architecture decisions we have made.
Taskflow is open-source at GitHub under MIT license and is being 
used by many academic and industrial projects~\cite{Taskflow}.

\section{Motivations}

Taskflow is motivated by our DARPA project to reduce the long
design times of modern circuits~\cite{IDEA}.
The main research objective is to advance \textit{computer-aided design} (CAD) tools
with heterogeneous parallelism to achieve
transformational performance and productivity milestones.
Unlike traditional loop-parallel scientific computing problems,
many CAD algorithms exhibit \textit{irregular computational patterns} 
and \textit{complex control flow} 
that require strategic task graph decompositions to benefit from 
heterogeneous parallelism~\cite{Huang_20_02}.
This type of complex parallel algorithm
is difficult to implement and execute efficiently using mainstream TGCS.
We highlight three reasons below,
\textit{end-to-end tasking}, \textit{GPU task graph parallelism},
and \textit{heterogeneous runtimes}.

\textbf{End-to-End Tasking} --
Optimization engines implement various graph and combinatorial algorithms
that frequently call for iterations, conditionals, and dynamic control flow.
Existing TGCSs~\cite{TBB, OpenMP, Nabbit, StarPU, Legion, PaRSEC, Kokkos, HPX, TPL},
closely rely on DAG models to define tasks and their dependencies.
Users implement control-flow decisions \textit{outside} the graph description 
via either statically unrolling the graph across fixed-length iterations 
or dynamically executing an ``if statement'' on the fly to decide the next path and so forth.
These solutions often incur rather complicated implementations
that lack \textit{end-to-end} parallelism using just one task graph entity.
For instance, when describing an iterative algorithm using a DAG model,
we need to repetitively wait for the task graph to complete at the end of each
iteration.
This wait operation is not cheap because it involves synchronization
between the application code and the TGCS runtime,
which could otherwise be totally avoided by supporting in-graph control-flow tasks.
More importantly, developers can benefit by making in-graph control-flow decisions
to efficiently overlap tasks both inside and outside control flow,
completely decided by a dynamic scheduler.
%


\textbf{GPU Task Graph Parallelism} --
Emerging GPU task graph acceleration, such as CUDA Graph~\cite{CUDAGraph},
can offer dramatic yet largely untapped performance advantages
by running a GPU task graph directly on a GPU.
This type of GPU task graph parallelism is particularly beneficial 
for many large-scale analysis and machine learning algorithms 
that compose thousands of dependent GPU operations
to run on the same task graph using iterative methods.
By creating an executable image for a GPU task graph, 
we can iteratively launch it with extremely low kernel overheads. 
However, 
existing TGCSs are short of a generic model to 
express and offload task graph parallelism 
directly on a GPU,
as opposed to a simple encapsulation of GPU operations into CPU tasks.

\textbf{Heterogeneous Runtimes} --
Many CAD algorithms compute extremely large circuit graphs. 
Different quantities are often dependent on each other, 
via either logical relation or physical net order, 
and are expensive to compute. 
The resulting task graph in terms of encapsulated function calls and task dependencies 
is usually very large. 
For example, the task graph representing a timing analysis on a million-gate design 
can add up to \textit{billions} of tasks that take several hours to finish~\cite{Huang_21_02}.
During the execution, tasks can run on CPUs or GPUs, or more frequently \textit{a mix}.
Scheduling these heterogeneously dependent tasks is a big challenge. 
Existing runtimes are good at either CPU- or GPU-focused work but rarely both simultaneously.

Therefore,
we argue that there is a critical need for a new heterogeneous task graph 
programming environment that supports in-graph control flow.
The environment must handle new scheduling challenges,
such as conditional dependencies and cyclic executions.
To this end, Taskflow aims to 
(1) introduce a new programming model that enables end-to-end expressions of
CPU-GPU dependent tasks along with algorithmic control flow and
(2) establish an efficient system runtime to support our model with high performance
across latency, energy efficiency, and throughput.
Taskflow focuses on a single heterogeneous node of CPUs and GPUs.


\section{Preliminary Results}

Taskflow is established atop our prior system,
\textit{Cpp-Taskflow}~\cite{Huang_21_02} which targets
CPU-only parallelism using a DAG model,
and extends its capability to heterogeneous computing using 
a new \textit{heterogeneous task dependency graph} (HTDG) programming model
beyond DAG.
Since we opened the source of Cpp-Taskflow/Taskflow,
it has been successfully adopted by much software, including important 
CAD projects~\cite{OpenRoad, Huang_20_01, ABCDPlace, Magical} 
under the DARPA ERI IDEA/POSH program~\cite{IDEA}.
Because of the success, 
we are recently invited to publish a 5-page TCAD brief to overview how Taskflow
address the parallelization challenges of CAD workloads~\cite{Huang_21_01}.
For the rest of the paper,
we will provide comprehensive details of the Taskflow system 
from the top-level programming model to the system runtime,
including several new technical materials for control-flow primitives, 
capturer-based GPU task graph parallelism, work-stealing algorithms and theory results,
and experiments.

\section{Taskflow Programming Model}

This section discusses five fundamental task types of Taskflow,
\textit{static task}, \textit{dynamic task}, \textit{module task},
\textit{condition task}, and \textit{cudaFlow} task.

\subsection{Static Tasking}

Static tasking is the most basic task type in Taskflow.
A static task takes a callable of no arguments and runs it.
The callable can be a generic C++ lambda function object,
binding expression, or a functor.
Listing~\ref{SimpleTDG} demonstrates a simple Taskflow program
of four static tasks, 
where \texttt{A} runs before \texttt{B} and \texttt{C}, 
and \texttt{D} runs after \texttt{B} and \texttt{C}.
The graph is run by an \textit{executor} 
which schedules dependent tasks across worker threads.
Overall, the code explains itself.

\begin{lstlisting}[language=C++,label=SimpleTDG,caption={A task graph of four static tasks.}]
tf::Taskflow taskflow;
tf::Executor executor;
auto [A, B, C, D] = taskflow.emplace(
  [] () { std::cout << "Task A"; },
  [] () { std::cout << "Task B"; },
  [] () { std::cout << "Task C"; },
  [] () { std::cout << "Task D"; } 
);
A.precede(B, C);  // A runs before B and C
D.succeed(B, C);  // D runs after  B and C
executor.run(tf).wait();
\end{lstlisting}

\subsection{Dynamic Tasking}

Dynamic tasking refers to the creation of a task graph during 
the execution of a task.
Dynamic tasks are spawned from a parent task and
are grouped to form a hierarchy called
\textit{subflow}.
Figure \ref{fig::DynamicTDG} shows an example of dynamic tasking.
The graph has four static tasks, \texttt{A}, \texttt{C}, 
\texttt{D}, and \texttt{B}.
The precedence constraints force \texttt{A} to run 
before \texttt{B} and \texttt{C},
and \texttt{D} to run after \texttt{B} and \texttt{C}.
During the execution of task \texttt{B},
it spawns another graph of three tasks, \texttt{B1}, \texttt{B2}, and \texttt{B3},
where \texttt{B1} and \texttt{B2} run before \texttt{B3}.
In this example, \texttt{B1}, \texttt{B2}, and \texttt{B3} are 
grouped to a subflow parented at \texttt{B}.

\begin{figure}[!h]
  \centering
  \includegraphics[width=.8\columnwidth]{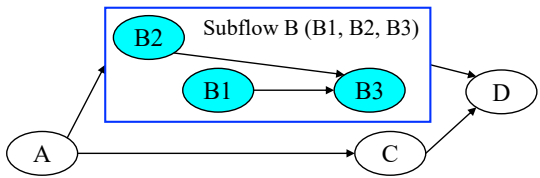}
  \caption{A task graph that spawns another task graph (\texttt{B1}, \texttt{B2}, and \texttt{B3}) during the execution of task \texttt{B}.}
  \label{fig::DynamicTDG}
\end{figure}

\begin{lstlisting}[language=C++,caption={Taskflow code of Figure~\ref{fig::DynamicTDG}.},label=listing::DynamicTDG]
auto [A, C, D] = taskflow.emplace(
  [] () { std::cout << "A"; },
  [] () { std::cout << "C"; },
  [] () { std::cout << "D"; }
);
auto B = tf.emplace([] (tf::Subflow& subflow) {
  std::cout << "B\n";
  auto [B1, B2, B3] = subflow.emplace(
    [] () { std::cout << "B1"; },
    [] () { std::cout << "B2"; },
    [] () { std::cout << "B3"; }
  );
  B3.succeed(B1, B2);
});
A.precede(B, C);
D.succeed(B, C);
\end{lstlisting}

Listing \ref{listing::DynamicTDG} shows the Taskflow code in
Figure \ref{fig::DynamicTDG}.
A dynamic task accepts a reference of type \texttt{tf::Subflow}
that is created by the executor during the execution of task \texttt{B}.
A subflow inherits
all graph building blocks of static tasking.
By default, a spawned subflow joins its parent task
(\texttt{B3} precedes its parent \texttt{B} implicitly),
forcing a subflow to follow the subsequent dependency constraints 
of its parent task.
Depending on applications, users can detach a subflow from its parent task
using the method \texttt{detach},
allowing its execution to flow independently.
A detached subflow will eventually join its parent taskflow.

\subsection{Composable Tasking}

Composable tasking enables developers to define task hierarchies 
and compose large task graphs from modular and reusable blocks
that are easier to optimize.
Figure \ref{fig::ComposableTasking} gives an example of 
a Taskflow graph using composition.
The top-level taskflow defines one static task \texttt{C} that runs before
a dynamic task \texttt{D} that spawns two dependent tasks \texttt{D1} and \texttt{D2}.
Task \texttt{D} precedes a \textit{module task} \texttt{E} that composes a taskflow of 
two dependent tasks \texttt{A} and \texttt{B}.

\begin{figure}[!h]
  \centering
  \centerline{\includegraphics[width=1.\columnwidth]{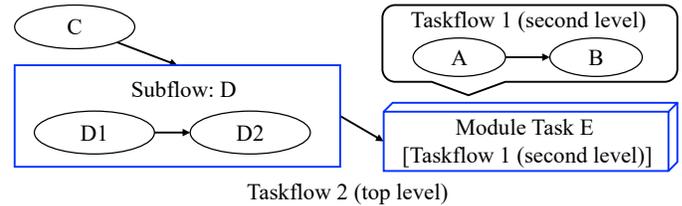}}
  \caption{An example of taskflow composition.}
  \label{fig::ComposableTasking}
\end{figure}

\begin{lstlisting}[language=C++,label=listing::ComposableTasking,caption={Taskflow code of Figure \ref{fig::ComposableTasking}.}]
// file 1 defines taskflow1
tf::Taskflow taskflow1;
auto [A, B] = taskflow1.emplace(
  [] () { std::cout << "TaskA"; },
  [] () { std::cout << "TaskB"; }
);  
A.precede(B);
// file 2 defines taskflow2
tf::Taskflow taskflow2;
auto [C, D] = taskflow2.emplace(
  [] () { std::cout << "TaskC"; },
  [] (tf::Subflow& sf) { 
    std::cout << "TaskD"; 
    auto [D1, D2] = sf.emplace(
      [] () { std::cout << "D1"; },
      [] () { std::cout << "D2"; }
    );  
    D1.precede(D2);
  }   
);  
auto E = taskflow2.composed_of(taskflow1); // module
D.precede(E);
C.precede(D);
\end{lstlisting}

Listing \ref{listing::ComposableTasking} shows the 
Taskflow code of Figure \ref{fig::ComposableTasking}.
It declares two taskflows, \texttt{taskflow1} and \texttt{taskflow2}.
\texttt{taskflow2} forms a module task \texttt{E} by calling
the method \texttt{composed\_of} from \texttt{taskflow1},
which is then preceded by task \texttt{D}.
Unlike a subflow task, a module task does not own the taskflow but maintains
a soft mapping to its composed taskflow.
Users can create multiple module tasks from the same taskflow 
but they must not run concurrently;
on the contrary, subflows are created dynamically and can run concurrently.
In practice, we use composable tasking to partition large parallel programs
into smaller or reusable taskflows in separate files 
(e.g., \texttt{taskflow1} in file 1 and \texttt{taskflow2} in file 2) 
to improve program modularity and testability.
Subflows are instead used for enclosing a task graph that needs stateful data referencing
via lambda capture.

%

\subsection{Conditional Tasking}

We introduce a new \textit{conditional tasking} model 
to overcome the limitation of existing frameworks in expressing \textit{general control flow}
beyond DAG.
A condition task is a callable that returns an integer index 
indicating the next successor task to execute.
The index is defined with respect to the order of the successors preceded
by the condition task.
Figure \ref{fig::conditional-tasking-if-else} shows an example of if-else 
control flow, and
Listing \ref{listing::conditional-tasking-if-else} gives its implementation.
The code is self-explanatory.
The condition task, \texttt{cond},
precedes two tasks, \texttt{yes} and \texttt{no}.
With this order,
if \texttt{cond} returns \texttt{0}, the execution moves on to \texttt{yes},
or \texttt{no} if \texttt{cond} returns \texttt{1}.

\begin{figure}[h]
  \centering
  \centerline{\includegraphics[width=.65\columnwidth]{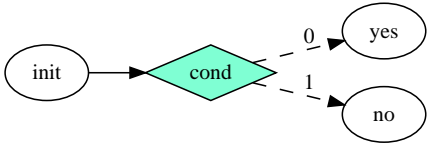}}
  \caption{A Taskflow graph of if-else control flow using one condition task (in diamond).}
  \label{fig::conditional-tasking-if-else}
\end{figure}

\begin{lstlisting}[language=C++,label=listing::conditional-tasking-if-else,caption={Taskflow program of Figure \ref{fig::conditional-tasking-if-else}.}]
auto [init, cond, yes, no] = taskflow.emplace(
 [] () { std::cout << "init"; },
 [] () { std::cout << "cond"; return 0; },
 [] () { std::cout << "cond returns 0"; },
 [] () { std::cout << "cond returns 1"; }
);
cond.succeed(init)
    .precede(yes, no);
\end{lstlisting}

Our condition task supports iterative control flow by introducing
a \textit{cycle} in the graph.
Figure \ref{fig::conditional-tasking-do-while} shows a task graph
of \textit{do-while} iterative control flow,
implemented in Listing \ref{listing::conditional-tasking-do-while}. 
The loop continuation condition is implemented by a single condition task, 
\texttt{cond},
that precedes two tasks, \texttt{body} and \texttt{done}.
When \texttt{cond} returns \texttt{0}, the execution loops back to \texttt{body}.
When \texttt{cond} returns \texttt{1}, the execution moves onto \texttt{done}
and stops.
In this example, we use only four tasks even though the control flow
spans 100 iterations.
Our model is more efficient and expressive than existing frameworks 
that count on dynamic tasking or recursive parallelism
to execute condition on the fly~\cite{Legion, PaRSEC}.

\begin{figure}[h]
  \centering
  \centerline{\includegraphics[width=.9\columnwidth]{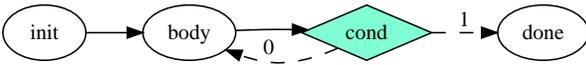}}
  \caption{A Taskflow graph of iterative control flow using one condition task.}
  \label{fig::conditional-tasking-do-while}
\end{figure}

\begin{lstlisting}[language=C++,label=listing::conditional-tasking-do-while,caption={Taskflow program of Figure \ref{fig::conditional-tasking-do-while}.}]
int i;
auto [init, body, cond, done] = taskflow.emplace(
  [&](){ i=0; },
  [&](){ i++; },
  [&](){ return i<100 ? 0 : 1; },
  [&](){ std::cout << "done"; }
);
init.precede(body);
body.precede(cond);
cond.precede(body, done);
\end{lstlisting}

Furthermore, our condition task can model non-deterministic control flow
where many existing models do not support.
Figure \ref{fig::conditional-tasking} 
shows an example of nested non-deterministic control flow 
frequently used in stochastic optimization (e.g., VLSI floorplan annealing~\cite{SAforVLSI}).
The graph consists of two regular tasks, \texttt{init} and \texttt{stop},
and three condition tasks, \texttt{F1}, \texttt{F2}, and \texttt{F3}.
Each condition task forms a dynamic control flow 
to randomly go to either the next task or loop back to \texttt{F1} 
with a probability of 1/2.
Starting from \texttt{init}, the expected number of condition tasks
to execute before reaching \texttt{stop} is eight.
Listing \ref{listing::conditional-tasking} implements 
Figure \ref{fig::conditional-tasking} in just 11 lines of code.

\begin{figure}[h]
  \centering
  \centerline{\includegraphics[width=1.\columnwidth]{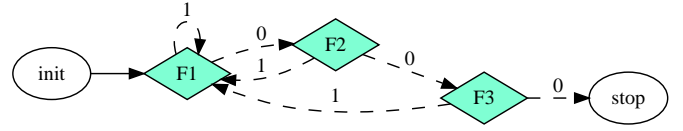}}
  \caption{A Taskflow graph of non-deterministic control flow using three condition tasks.}
  \label{fig::conditional-tasking}
\end{figure}

\begin{lstlisting}[language=C++,label=listing::conditional-tasking,caption={Taskflow program of Figure \ref{fig::conditional-tasking}.}]
auto [init, F1, F2, F3, stop] = taskflow.emplace(
  [] () { std::cout << "init"; },
  [] () { return rand()%2 },
  [] () { return rand()%2 },
  [] () { return rand()%2 },
  [] () { std::cout << "stop"; }
);
init.precede(F1);
F1.precede(F2, F1);
F2.precede(F3, F1);
F3.precede(stop, F1);
\end{lstlisting}

The advantage of our conditional tasking is threefold.
First, it is simple and expressive. 
Developers benefit from the ability to make \textit{in-graph} control-flow decisions 
that are integrated within task dependencies.
This type of decision making is different from dataflow~\cite{HPX}
as we do not abstract data but tasks,
and is more general than the primitive-based method~\cite{Yu_18_01}
that is limited to domain applications.
Second,
condition tasks can be associated with other tasks 
to integrate control flow into a unified graph entity.
Users ought not to partition the control flow or unroll it to a flat DAG,
but focus on expressing dependent tasks and control flow.
The later section will explain our scheduling algorithms for condition tasks.
Third,
our model enables
developers to efficiently overlap tasks both inside and outside
control flow.
For example, Figure \ref{fig::conditional-tasking-complex}
implements a task graph of three control-flow blocks,
and \texttt{cond\_1} can run in parallel with
\texttt{cond\_2} and \texttt{cond\_3}.
This example requires only 30 lines of code.

\begin{figure}[h]
  \centering
  \centerline{\includegraphics[width=1.\columnwidth]{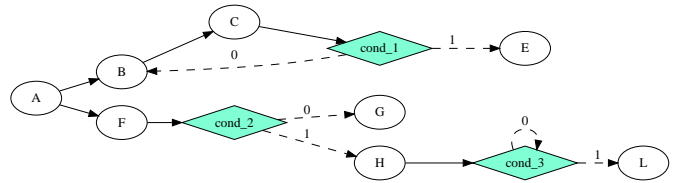}}
  \caption{A Taskflow graph of parallel control-flow blocks using three condition tasks.}
  \label{fig::conditional-tasking-complex}
\end{figure}

\subsection{Heterogeneous Tasking}

We introduce a new heterogeneous task graph programming model
by leveraging C++ closure and emerging GPU task graph acceleration, 
\textit{CUDA Graph}~\cite{CUDAGraph}.
Figure \ref{fig::saxpy} and Listing \ref{listing::saxpy} show the
canonical CPU-GPU saxpy (A•X plus Y) workload
and its implementation using our model.
Our model lets users describe a GPU workload in a \textit{task graph} called
\textit{cudaFlow} rather than aggregated GPU operations using explicit
CUDA streams and events.
A cudaFlow lives inside a closure and defines methods for constructing a GPU task graph.
In this example,
we define two parallel CPU tasks (\texttt{allocate\_x}, \texttt{allocate\_y})
to allocate unified shared memory (\texttt{cudaMallocManaged})
and one \texttt{cudaFlow} task to spawn a GPU task graph consisting
of two host-to-device (H2D) transfer tasks
(\texttt{h2d\_x}, \texttt{h2d\_y}), 
one saxpy kernel task (\texttt{kernel}), 
and two device-to-host (D2H) transfer tasks 
(\texttt{d2h\_x}, \texttt{d2h\_y}),
in this order of task dependencies.
Task dependencies are established through \texttt{precede} or \texttt{succeed}.
Apparently, \texttt{cudaFlow} must run after
\texttt{allocate\_x} and \texttt{allocate\_y}.
We emplace this cudaFlow on GPU 1 (\texttt{emplace\_on}).
When defining cudaFlows on specific GPUs,
users are responsible for ensuring all involved memory operations
stay in valid GPU contexts.

\begin{figure}[h]
  \centering
  \centerline{\includegraphics[width=1.\columnwidth]{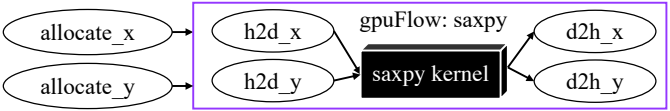}}
  \caption{A saxpy (``single-precision A·X plus Y") task graph using two CPU tasks and one cudaFlow task.}
  \label{fig::saxpy}
\end{figure}

\begin{lstlisting}[language=C++,label=listing::saxpy,caption={Taskflow program of Figure \ref{fig::saxpy}.}]
__global__ void saxpy(int n,int a,float *x,float *y);

const unsigned N = 1<<20;
std::vector<float> hx(N, 1.0f), hy(N, 2.0f);
float *dx{nullptr}, *dy{nullptr};

auto [allocate_x, allocate_y] = taskflow.emplace(
  [&](){ cudaMallocManaged(&dx, N*sizeof(float));}
  [&](){ cudaMallocManaged(&dy, N*sizeof(float));}
);
auto cudaFlow = taskflow.emplace_on(
  [&](tf::cudaFlow& cf) {
    auto h2d_x = cf.copy(dx, hx.data(), N);
    auto h2d_y = cf.copy(dy, hy.data(), N);
    auto d2h_x = cf.copy(hx.data(), dx, N);
    auto d2h_y = cf.copy(hy.data(), dy, N);
    auto kernel = cf.kernel(
      GRID, BLOCK, SHM, saxpy, N, 2.0f, dx, dy
    );
    kernel.succeed(h2d_x, h2d_y)
          .precede(d2h_x, d2h_y);
  }, 1
);
cudaFlow.succeed(allocate_x, allocate_y);
\end{lstlisting}

Our cudaFlow has the three key motivations.
First, users focus on the graph-level expression of dependent GPU operations
without wrangling with low-level streams.
They can easily visualize the graph by Taskflow to reduce turnaround time.
Second, our closure \textit{forces} users to express their intention
on what data storage mechanism should be used for each captured variable.
For example, Listing \ref{listing::saxpy} captures all data
(e.g., \texttt{hx}, \texttt{dx}) in \textit{reference} 
to form a \textit{stateful closure}.
When \texttt{allocate\_x} and \texttt{allocate\_y} finish,
the cudaFlow closure can access the correct state of \texttt{dx} and \texttt{dy}.
This property is very important for heterogeneous graph parallelism
because CPU and GPU tasks need to share states of data to collaborate with each other.
Our model makes it easy and efficient to capture data regardless of its scope.
Third, by abstracting GPU operations to a task graph closure, 
we judiciously hide implementation details for portable optimization.
By default, a cudaFlow maps to a CUDA graph that can be executed 
using a single CPU call.
On a platform that does not support CUDA Graph,
we fall back to a stream-based execution.
%

Taskflow does not dynamically choose whether to execute tasks on CPU or GPU, 
and does not manage GPU data with another abstraction.
This is a software decision we have made when designing cudaFlow
based on our experience in parallelizing CAD using existing TGCSs. 
While it is always interesting to see what abstraction is best suited for which application, 
in our field, developing high-performance CAD algorithms requires many custom
efforts on optimizing the memory and data layouts~\cite{Huang_20_02, Huang_21_03}. 
Developers tend to do this statically in their own hands, 
such as direct control over raw pointers and explicit memory placement on a GPU,
while leaving tedious details of runtime load balancing to a dynamic scheduler. 
After years of
research, we have concluded to not abstract memory or data because they are 
application-dependent. 
This decision allows Taskflow to be \textit{framework-neutral} while enabling application
code to take full advantage of native or low-level GPU programming toolkits.
%
%

\begin{lstlisting}[language=C++,label=listing::saxpy-capturer,caption={Taskflow program of Figure \ref{fig::saxpy} using a capturer.}]
taskflow.emplace_on([&](tf::cudaFlowCapturer& cfc) {
  auto h2d_x = cfc.copy(dx, hx.data(), N);
  auto h2d_y = cfc.copy(dy, hy.data(), N);
  auto d2h_x = cfc.copy(hx.data(), dx, N);
  auto d2h_y = cfc.copy(hy.data(), dy, N);
  auto kernel = cfc.on([&](cudaStream_t s){
    invoke_3rdparty_saxpy_kernel(s);
  });
  kernel.succeed(h2d_x, h2d_y)
        .precede(d2h_x, d2h_y);
}, 1);
\end{lstlisting}

Constructing a GPU task graph using cudaFlow requires all kernel parameters
are known in advance.
However, third-party applications, such as cuDNN and cuBLAS,
do not open these details but provide an API for users to invoke hidden kernels
through custom streams.
The burden is on users to decide a stream layout and witness its
concurrency across dependent GPU tasks.
To deal with this problem,
we design a cudaFlow \textit{capturer} to capture GPU tasks
from existing stream-based APIs.
Listing \ref{listing::saxpy-capturer} outlines an implementation
of the same saxpy task graph
in Figure \ref{fig::saxpy} using a cudaFlow capturer,
assuming the saxpy kernel is only invocable through a stream-based API.

%

Both cudaFlow and cudaFlow capturer can work seamlessly with condition tasks.
Control-flow decisions frequently happen at the boundary
between CPU and GPU tasks.
For example, a heterogeneous $k$-means algorithm iteratively
uses GPU to accelerate the finding of $k$ centroids and then
uses CPU to check if the newly found centroids converge to application rules.
Taskflow enables an end-to-end expression of such a workload in a single graph entity,
as shown in Figure \ref{fig::conditional-tasking-kmeans}
and Listing \ref{listing::conditional-tasking-kmeans}.
This capability largely improves the efficiency of modeling complex CPU-GPU workloads,
and our scheduler can dynamically overlap CPU and GPU tasks
across different control-flow blocks.

\begin{figure}[h]
  \centering
  \centerline{\includegraphics[width=1.\columnwidth]{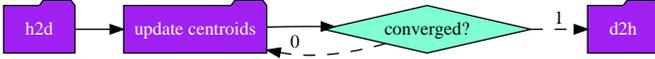}}
  \caption{A cyclic task graph using three cudaFlow tasks and one condition task to model 
  an iterative $k$-means algorithm.}
  \label{fig::conditional-tasking-kmeans}
\end{figure}

\begin{lstlisting}[language=C++,label=listing::conditional-tasking-kmeans,caption={Taskflow program of Figure \ref{fig::conditional-tasking-kmeans}.}]
auto [h2d, update, cond, d2h] = taskflow.emplace(
  [&](tf::cudaFlow& cf){ /* copy input to GPU */ },
  [&](tf::cudaFlow& cf){ /* update kernel */ },
  [&](){ return converged() ? 1 : 0; },
  [&](tf::cudaFlow& cf){ /* copy result to CPU */ }
);
h2d.precede(update);
update.precede(cond);
cond.precede(update, d2h);
\end{lstlisting}

\section{Taskflow System Runtime}

Taskflow enables users to express CPU-GPU dependent tasks
that integrate control flow into an HTDG.
To support our model with high performance,
we design the system runtime at two scheduling levels,
\textit{task level} and \textit{worker level}.
The goal of task-level scheduling is to (1) devise a feasible, efficient
execution for in-graph control flow and
(2) transform each GPU task into a runnable instance on a GPU.
The goal of worker-level scheduling is to 
optimize the execution performance by dynamically balancing
the worker count with task parallelism.

\subsection{Task-level Scheduling Algorithm}

\subsubsection{Scheduling Condition Tasks}
Conditional tasking is powerful but challenging to schedule.
Specifically, 
we must deal with conditional dependency and cyclic execution 
without encountering \textit{task race}, 
i.e., only one thread can touch a task at a time.
More importantly,
we need to let users easily understand our task scheduling flow
such that they can infer if a written task graph is properly
conditioned and schedulable.
To accommodate these challenges, 
we separate the execution logic 
between condition tasks and other tasks using two dependency notations, 
\textit{weak dependency} (out of condition tasks) and 
\textit{strong dependency} (other else).
For example, the six dashed arrows in Figure \ref{fig::conditional-tasking}
are weak dependencies
and the solid arrow \texttt{init}$\rightarrow$\texttt{F1}
is a strong dependency.
Based on these notations,
we design a simple and efficient algorithm for scheduling tasks,
as depicted in Figure \ref{fig::task-scheduling}.
When the scheduler receives an HTDG,
it (1) starts with tasks of \textit{zero} dependencies (both strong and weak) 
and continues executing tasks whenever \textit{strong} remaining dependencies are met,
or (2) skips this rule for weak dependency and 
directly jumps to the task indexed by the return of that condition task.
%

\begin{figure}[h]
  \centering
  \centerline{\includegraphics[width=.8\columnwidth]{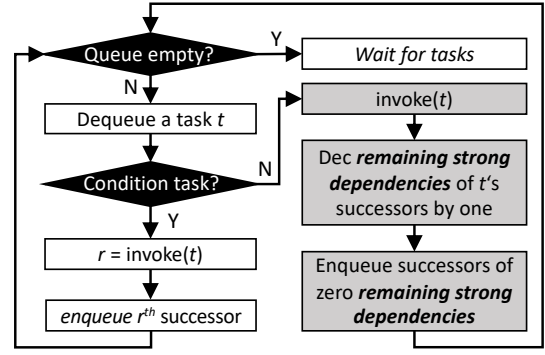}}
  \caption{Flowchart of our task scheduling.}
  \label{fig::task-scheduling}
\end{figure}


Taking Figure \ref{fig::conditional-tasking} for example,
the scheduler starts with \texttt{init} (zero weak and strong dependencies)
and proceeds to \texttt{F1}.
Assuming \texttt{F1} returns \texttt{0}, 
the scheduler proceeds to its first successor, \texttt{F2}.
Now, assuming \texttt{F2} returns \texttt{1}, 
the scheduler proceeds to its second successor, \texttt{F1},
which forms a cyclic execution and so forth.
With this concept, the scheduler will cease at \texttt{stop} when \texttt{F1},
\texttt{F2}, and \texttt{F3} all return \texttt{0}.
Based on this scheduling algorithm,
users can quickly infer whether their task graph defines correct control flow.
For instance, adding a strong dependency from \texttt{init} to \texttt{F2} 
may cause task race on \texttt{F2}, due to two execution paths,
\texttt{init}$\rightarrow$\texttt{F2} and \texttt{init}$\rightarrow$\texttt{F1}$\rightarrow$\texttt{F2}.

Figure \ref{fig::conditional-tasking-pitfall} shows 
two common pitfalls of conditional tasking,
based on our task-level scheduling logic.
The first example has no source for the scheduler to start with.
A simple fix is to add a task \texttt{S} of zero dependencies.
The second example may race on \texttt{D}, if \texttt{C} returns 0
at the same time \texttt{E} finishes.
A fix is to partition the control flow at \texttt{C} and \texttt{D}
with an auxiliary node \texttt{X}
such that \texttt{D} is strongly conditioned by \texttt{E} and \texttt{X}.

\begin{figure}[h]
  \centering
  \centerline{\includegraphics[width=1.\columnwidth]{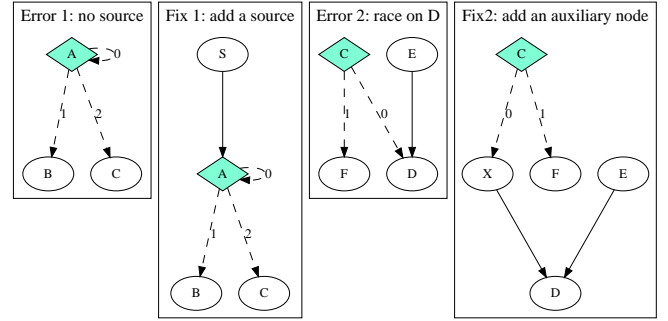}}
  \caption{Common pitfalls of conditional tasking.}
  \label{fig::conditional-tasking-pitfall}
\end{figure}

\subsubsection{Scheduling GPU Tasks}

We leverage modern \textit{CUDA Graph}~\cite{CUDAGraph}
to schedule GPU tasks.
CUDA graph is a new asynchronous task graph programming model 
introduced in CUDA 10
to enable more efficient launch and execution of GPU work than streams.
There are two types of GPU tasks, cudaFlow and cudaFlow capturer.
For each scheduled cudaFlow task,
since we know all the operation parameters,
we construct a CUDA graph that maps each task in the cudaFlow,
such as copy and kernel,
and each dependency to a node and an edge in the CUDA graph.
Then, we submit it to the CUDA runtime for execution.
This organization is simple and efficient, 
especially under modern GPU architectures (e.g., Nvidia Ampere) 
that support hardware-level acceleration for graph parallelism.

On the other hand,
for each scheduled cudaFlow capturer task,
our runtime transforms the captured GPU tasks and dependencies 
into a CUDA graph using \textit{stream capture}~\cite{CUDAGraph}.
The objective is to decide a stream layout optimized for kernel concurrency
without breaking task dependencies.
We design a greedy round-robin algorithm 
to transform a cudaFlow capturer to a CUDA graph,
as shown in Algorithm \ref{alg::optimize}.
Our algorithm starts by \textit{levelizing} the capturer graph into a 
two-level array of tasks in their topological orders.
Tasks at the same level can run simultaneously.
However, assigning each independent task here a unique stream does not
produce decent performance, because GPU has a limit on
the maximum kernel concurrency (e.g., 32 for RTX 2080).
We give this constraint to users as a tunable parameter, $max\_streams$.
We assign each levelized task an $id$ equal to its index in the array at its level.
%
Then,
we can quickly assign each task a stream using the round-robin arithmetic
(line \ref{alg::optimize::modulo}).
Since tasks at different levels have dependencies,
we need to record an event 
(lines \ref{alg::optimize::recordloopbeg}:\ref{alg::optimize::recordloopend}) 
and wait on the event
(lines \ref{alg::optimize::waitloopbeg}:\ref{alg::optimize::waitloopend})
from both sides of a dependency,
saved for those issued in the same stream
(line \ref{alg::optimize::waitloopexcept} and 
line \ref{alg::optimize::recordloopexcept}).

\begin{algorithm}[h]
 \KwIn{a cudaFlow capturer $C$}
 \KwOut{a transformed CUDA graph $G$}
 \SetKw{KwBreak}{break}
 \SetKw{KwTrue}{true}
 \SetKw{KwNot}{not}
 \SetKw{KwOr}{or}
 \SetKwRepeat{Do}{do}{while}
 \BlankLine
 $S \leftarrow$ get\_capture\_mode\_streams($max\_streams$)\;
 $L \leftarrow$ levelize($C$)\;
 $l \leftarrow L.min\_level$\;
 \While{$l <= L.max\_level$} {  \label{alg::optimize::level_loop}
   \ForEach{$t \in$ L.get\_tasks(l)}{ 
     $s \leftarrow (t.id \mod max\_streams) $\; \label{alg::optimize::modulo}
     \ForEach{$p \in t.predecessors$} { \label{alg::optimize::waitloopbeg}
       \If{$s \neq (p.id \mod max\_streams)$} { \label{alg::optimize::waitloopexcept}
         stream\_wait\_event($S[s]$, $p.event$)\;
       }
     } \label{alg::optimize::waitloopend}
     stream\_capture($t$, $S[s]$)\;
     \ForEach{$n \in t.successors$} { \label{alg::optimize::recordloopbeg}
       \If{$s \neq (n.id \mod max\_streams)$} { \label{alg::optimize::recordloopexcept}
         stream\_record\_event($S[s]$, $p.event$)\;
       }
     } \label{alg::optimize::recordloopend}
   }
 }
 $G \leftarrow$ end\_capture\_mode\_streams($S$)\;
 \Return $G$\;
 \caption{make\_graph($G$)}
 \label{alg::optimize}
\end{algorithm}

\begin{figure}[!h]
  \centering
  \includegraphics[width=1.\columnwidth]{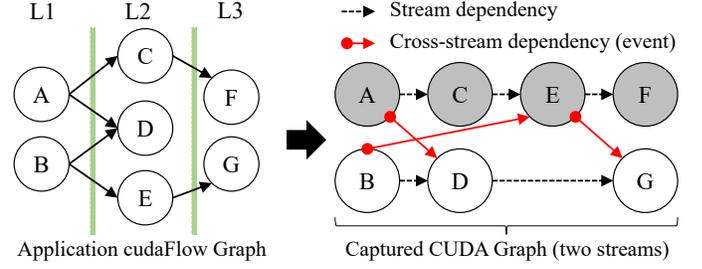}
  \caption{Illustration of Algorithm \ref{alg::optimize} on transforming an application cudaFlow capturer graph into a native CUDA graph using two streams.}
  \label{fig::capturer}
\end{figure}

Figure \ref{fig::capturer} gives an example of transforming a user-given cudaFlow capturer graph
into a native CUDA graph using two streams (i.e., $max\_stream = 2$) for execution.
The algorithm first levelizes the graph by performing a topological traversal
and assign each node an id equal to its index at the level.
For example, \texttt{A} and \texttt{B} are assigned 0 and 1,
\texttt{C}, \texttt{D}, and \texttt{E} are assigned 0, 1, and 2, and so on.
These ids are used to quickly determine the mapping between a stream and a node
in our round-robin loop,
because CUDA stream only allows inserting events from the latest node in the queue.
For instance,
when \texttt{A} and \texttt{B} are assigned to stream 0 (upper row) and stream 1 (lower row)
during the level-by-level traversal (line \ref{alg::optimize::level_loop} of
Algorithm \ref{alg::optimize}),
we can determine ahead of the stream numbers of their successors 
and find out the two cross-stream dependencies, \texttt{A}$\rightarrow$\texttt{D} 
and \texttt{B}$\rightarrow$\texttt{E}, that need recording events.
Similarly, we can wait on recorded events by scanning the predecessors
of each node to find out cross-stream event dependencies.

\subsection{Worker-level Scheduling Algorithm}

At the worker level,
we leverage \textit{work stealing} to execute submitted tasks with dynamic load balancing.
Work stealing has been extensively studied in multicore 
programming~\cite{TBB, Nabbit, A-STEAL, ABP, BWS, Cilk++, TPL, EWS, X10},
but an efficient counterpart for hybrid CPU-GPU or more general heterogeneous systems 
remains demanding.
This is a challenging research topic, especially under Taskflow's HTDG model.
When executing an HTDG,
a CPU task can submit both CPU and GPU tasks and vice versa 
whenever dependencies are met.
The available task parallelism changes dynamically,
and there are no ways to predict the next coming tasks under dynamic control flow.
To achieve good system performance,
the scheduler must balance the number of worker threads
with dynamically generated tasks 
to control the number of \textit{wasteful steals} 
because the wasted resources should have been used 
by useful workers or other concurrent programs~\cite{A-STEAL, BWS}.

Keeping workers busy in awaiting tasks with a \textit{yielding} mechanism 
is a commonly used work-stealing framework~\cite{ABP, StarPU, XKAAPI}.
However, this approach is not cost-efficient, 
because it can easily over-subscribe resources when tasks become scarce, 
especially around the decision-making points of control flow.
The sleep-based mechanism is another way 
to suspend the workers frequently failing in steal attempts.
A worker is put into sleep by waiting for a condition variable 
to become true.
When the worker sleeps, 
OS can grant resources to other workers for running useful jobs.
Also, reducing wasteful steals can improve both
the inter-operability of a concurrent program and 
the overall system performance, including latency, throughput, and energy efficiency
to a large extent~\cite{BWS}.
Nevertheless, deciding 
\textit{when and how to put workers to sleep,
wake up workers to run, and balance the numbers of workers with
dynamic task parallelism}
is notoriously challenging to design correctly and implement efficiently.

Our previous work~\cite{Lin_20_01} has introduced an adaptive
work-stealing algorithm to address a similar line of the challenge 
yet in a CPU-only environment
by maintaining a loop invariant between active and idle workers.
However,
extending this algorithm to a heterogeneous target is not easy,
because we need to consider the adaptiveness in different heterogeneous domains
and bound the total number of wasteful steals across all domains 
at any time of the execution.
To overcome this challenge,
we introduce a new scheduler architecture
and an adaptive worker management algorithm 
that are both generalizable to arbitrary heterogeneous domains.
We shall prove the proposed work-stealing algorithm can deliver a strong upper bound
on the number of wasteful steals at any time during the execution.

\subsubsection{Heterogeneous Work-stealing Architecture}

At the architecture level, our scheduler maintains a set of workers for each task domain
(e.g., CPU, GPU).
A worker can only steal tasks of the same domain from others.
Figure \ref{fig::architecture} shows the architecture of our work-stealing scheduler 
on two domains, CPU and GPU.
By default, the number of domain workers equals the number of domain devices
(e.g., CPU cores, GPUs).
We associate each worker with two separate task queues,
a CPU task queue (CTQ) and a GPU task queue (GTQ),
and declare a pair of CTQ and GTQ shared by all workers.
The shared CTQ and GTQ pertain to the scheduler
and are primarily used for external threads to submit HTDGs.
A CPU worker can push and pop a new task into and from 
its local CTQ, 
and can steal tasks from all the other CTQs;
the structure is symmetric to GPU workers.
This separation allows a worker to quickly 
insert dynamically generated tasks to their corresponding queues
without contending with other workers.

\begin{figure}[h]
  \centering
  \centerline{\includegraphics[width=1.\columnwidth]{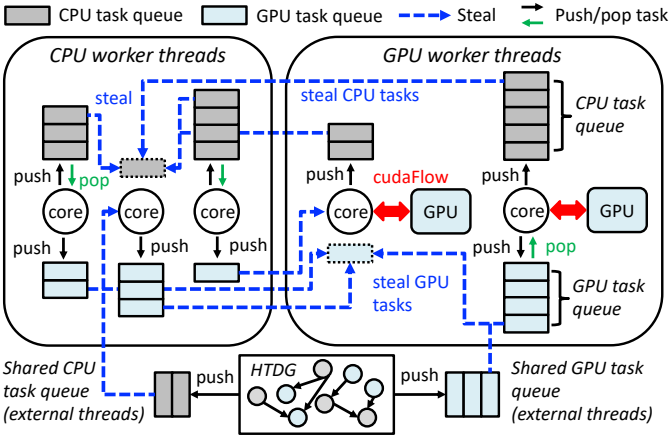}}
  \caption{Architecture of our work-stealing scheduler on two domains, CPU and GPU.}
  \label{fig::architecture}
\end{figure}

We leverage two existing concurrent data structures,
\textit{work-stealing queue} and \textit{event notifier},
to support our scheduling architecture.
We implemented the task queue based on the lock-free algorithm
proposed by~\cite{WSQ}.
Only the queue owner can pop/push a task from/into one end of the queue,
while multiple threads can steal a task from the other end at the same time.
Event notifier is a two-phase commit protocol (2PC) that
allows a worker to wait on a binary predicate 
in a \textit{non-blocking} fashion~\cite{2PC}.
The idea is similar to the 2PC in distributed systems and computer networking.
The waiting worker first checks the predicate and calls \texttt{prepare\_wait} 
if it evaluates to false.
The waiting worker then checks the predicate again and calls \texttt{commit\_wait}
to wait, if the outcome remains false, 
or \texttt{cancel\_wait} to cancel the request.
Reversely, the notifying worker changes the predicate to true
and call \texttt{notify\_one} or \texttt{notify\_all}
to wake up one or all waiting workers.
Event notifier is particularly useful for our scheduler architecture
because we can keep notification between workers non-blocking.
We develop one event notifier for each domain, 
based on Dekker's algorithm by~\cite{2PC}.

\subsubsection{Heterogeneous Work-stealing Algorithm}
\label{sec::HeterogeneousWorkStealingAlgorithm}

Atop this architecture, 
we devise an efficient algorithm to \textit{adapt} the number of active workers 
to dynamically generated tasks such that
threads are not underutilized when tasks are abundant nor
overly subscribed when tasks are scarce.
Our adaptiveness is different from existing frameworks, 
such as constant wake-ups~\cite{TBB, BWS},
data locality~\cite{LAWS, DeterministicWS},
and watchdogs~\cite{BWS}.
Instead, we extend our previous work~\cite{Lin_20_01} to keep a \textit{per-domain invariant} 
to control the numbers of thieves and,
consequently, wasteful steals
based on the active worker count: 
\textit{When an active worker exists, we keep at least one worker making steal attempts unless all workers are active.}

Unlike the CPU-only scheduling environment in~\cite{Lin_20_01}, 
the challenge to keep this invariant
in a heterogeneous target comes from the heterogeneously dependent tasks and
cross-domain worker notifications,
as a CPU task can spawn a GPU task and vice versa.
Our scheduler architecture is particularly designed to tackle this challenge
by separating decision controls to a per-domain basis.
This design allows us to realize the invariant via an adaptive strategy--\textit{the 
last thief to become active will wake up
a worker in the same domain to take over its thief role, and so forth}.
External threads (non-workers) submit tasks through the 
shared task queues and wake up workers to run tasks.

%

\begin{algorithm}[h]
 \SetKw{KwNIL}{NIL}
 \SetKwInput{KwPerWorkerGlobal}{Per-worker global}
 \KwIn{$w$: a worker}
 \KwPerWorkerGlobal{$t$: a task (initialized to \KwNIL)}
 \SetKw{KwBreak}{break}
 \While{true} {
   exploit\_task($w$, $t$)\;
   \If{wait\_for\_task(w, t) == false} {  \label{alg::worker_loop::wait_for_task}
     \KwBreak\;
   }
 }
 \caption{worker\_loop($w$)}
 \label{alg::worker_loop}
\end{algorithm}

\begin{algorithm}[h]
 \SetKwInput{KwPerWorkerGlobal}{Per-worker global}
 \KwIn{$w$: a worker (domain $d_w$)}
 \KwPerWorkerGlobal{$t$: a task}
 \SetKw{KwBreak}{break}
 \SetKw{KwTrue}{true}
 \SetKw{KwNot}{not}
 \SetKw{KwAnd}{and}
 \SetKw{KwOr}{or}
 \SetKw{KwNIL}{NIL}
 \SetKwRepeat{Do}{do}{while}
 \If{$t \neq$ \KwNIL} {
   \If{AtomInc(actives[$d_w$]) == 1 \KwAnd thieves[$d_w$] == 0} {  \label{alg::exploit_task::notify_beg}
     $notifier[d_w]$.notify\_one()\;
   }  \label{alg::exploit_task::notify_end}
   \Do{t $\neq$ \KwNIL} {  \label{alg::exploit_task::drain_beg}
     execute\_task($w$, $t$)\;
     $t \leftarrow w.task\_queue[d_w]$.pop()\;
   }  \label{alg::exploit_task::drain_end}
   AtomDec($actives[d_w]$)\;  \label{alg::exploit_task::dec_actives}
 }
 \caption{exploit\_task($w$, $t$)}
 \label{alg::exploit_task}
\end{algorithm}

Our scheduling algorithm is symmetric by domain.
Upon spawned,
each worker enters the loop in Algorithm \ref{alg::worker_loop}.
Each worker has a per-worker global pointer $t$ to a task that is either 
stolen from others or popped out from the worker's local task queue
after initialization; the notation will be used in the rest of algorithms.
The loop iterates two functions,
\texttt{exploit\_task} and \texttt{wait\_for\_task}.
Algorithm \ref{alg::exploit_task} implements the function
\texttt{exploit\_task}.
We use two scheduler-level arrays of atomic variables, 
\texttt{actives} 
and \texttt{thieves},
to record for each domain the number of workers that 
are actively running tasks
and the number of workers that are making steal attempts, respectively.\footnote{While 
our pseudocodes use array notations of atomic variables for the sake of brevity,
the actual implementation considers padding to avoid false-sharing effects.
}
Our algorithm relies on these atomic variables
to decide when to put a worker to sleep 
for reducing resource waste
and 
when to bring back a worker for running new tasks.
Lines \ref{alg::exploit_task::notify_beg}:\ref{alg::exploit_task::notify_end}
implement our adaptive strategy using two lightweight atomic operations.
In our pseudocodes, the two atomic operations, \textit{AtomInc} and \textit{AtomDec},
return the results after incrementing and decrementing the values by one, respectively.
Notice that the order of these two comparisons matters
(i.e., active workers and then thieves),
as they are used to synchronize with other workers
in the later algorithms.
Lines \ref{alg::exploit_task::drain_beg}:\ref{alg::exploit_task::drain_end}
drain out the local task queue
and executes all the tasks using 
\texttt{execute\_task} in Algorithm \ref{alg::execute_task}.
Before leaving the function, the worker decrements
\texttt{actives} by one (line \ref{alg::exploit_task::dec_actives}).

\begin{algorithm}[h]
 \SetKwInput{KwPerWorkerGlobal}{Per-worker global}
 \KwIn{$w$: a worker}
 \KwPerWorkerGlobal{$t$: a task}
 \SetKw{KwAnd}{and}
 \SetKw{KwTrue}{true}
 \SetKw{KwFalse}{false}

 $r \leftarrow$ invoke\_task\_callable($t$)\; \label{alg::execute_task::visitor}
 
 \If{r.has\_value()}{  \label{alg::execute_task::condition_task_beg}
   submit\_task($w$, $t.successors[r]$)\;
   \Return\;
 }  \label{alg::execute_task::condition_task_end}
 \ForEach{s $\in$ t.successors} {  \label{alg::execute_task::others_beg}
   \If{AtomDec(s.strong\_dependents) == 0} {
     submit\_task($w$, $s$)\;
   }
 }  \label{alg::execute_task::others_end}

 \caption{execute\_task($w$, $t$)}
 \label{alg::execute_task}
\end{algorithm}

\begin{algorithm}[h]
 \SetKwInput{KwPerWorkerGlobal}{Per-worker global}
 \KwIn{$w$: a worker (domain $d_w$)}
 \KwPerWorkerGlobal{$t$: a task (domain $d_t$)}
 \SetKw{KwAnd}{and}
 \SetKw{KwTrue}{true}
 \SetKw{KwFalse}{false}
 \SetKw{KwReturn}{return}

 $w.task\_queue[d_t]$.push($t$)\; \label{alg::submit_task::push_task}
   \If{$d_w != d_t$} { \label{alg::submit_task::domain_check}
     \If{actives[$d_t$] == 0 \KwAnd thieves[$d_t$] == 0} {  \label{alg::submit_task::notify_beg}
       $notifier[d_t]$.notify\_one()\;  \label{alg::submit_task::notify}
     }  \label{alg::submit_task::notify_end}
 }

 \caption{submit\_task($w$, $t$)}
 \label{alg::submit_task}
\end{algorithm}

Algorithm \ref{alg::execute_task} implements the function \texttt{execute\_task}.
We invoke the callable of the task 
(line \ref{alg::execute_task::visitor}).
If the task returns a value (i.e., a condition task),
we directly submit the task of the indexed successor 
(lines \ref{alg::execute_task::condition_task_beg}:\ref{alg::execute_task::condition_task_end}).
Otherwise, 
we remove the task dependency from all immediate successors and
submit new tasks of zero remaining strong dependencies 
(lines \ref{alg::execute_task::others_beg}:\ref{alg::execute_task::others_end}).
The detail of submitting a task
is shown in Algorithm \ref{alg::submit_task}.
The worker inserts the task into the queue of the corresponding domain
(line \ref{alg::submit_task::push_task}).
If the task does not belong to the worker's domain (line \ref{alg::submit_task::domain_check}),
the worker wakes up one worker from that domain
if there are no active workers or thieves 
(lines \ref{alg::submit_task::notify_beg}:\ref{alg::submit_task::notify_end}). 
The function \texttt{submit\_task} is internal to the workers
of a scheduler.
External threads never touch this call.

\begin{algorithm}[!h]
 \SetKwInput{KwPerWorkerGlobal}{Per-worker global}
 \KwIn{$w$: a worker (domain $d_w$)}
 \KwPerWorkerGlobal{$t$: a task}
 \KwOut{a boolean signal of stop}
 \SetKw{KwGoTo}{goto}
 \SetKw{KwAnd}{and}
 \SetKw{KwTrue}{true}
 \SetKw{KwFalse}{false}
 \SetKw{KwNIL}{NIL}

 AtomInc($thieves[d_w]$)\;

 explore\_task($w, t$)\;  \label{alg::wait_for_task::explore_task}

 \If{t $\neq$ \KwNIL} {  \label{alg::wait_for_task::thief_to_active_beg}
   \If{AtomDec(thieves[$d_w$]) == 0}{
     $notifier[d_w]$.notify\_one()\;
   }
   \Return \KwTrue\;
 }  \label{alg::wait_for_task::thief_to_active_end}

 $notifier[d_w]$.prepare\_wait($w$)\;  \label{alg::wait_for_task::prepare_wait}

 \If{task\_queue[$d_w$].empty() $\neq$ \KwTrue} {  \label{alg::wait_for_task::shared_empty_beg}
   $notifier[d_w]$.cancel\_wait($w$)\;   \label{alg::wait_for_task::cancel_wait}
   $t \leftarrow task\_queue[d_w]$.steal()\;  \label{alg::wait_for_task::steal_after_cancel_wait}
   \If{$t \neq \KwNIL$} {  \label{alg::wait_for_task::shared_empty_last_thief_beg}
     \If{AtomDec($thieves[d_w]$) == 0}{
       $notifier[d_w]$.notify\_one()\;
     }
     \Return \KwTrue\;
   }  \label{alg::wait_for_task::shared_empty_last_thief_end}
   \KwGoTo Line 2\;  \label{alg::wait_for_task::goto_steal_from_shared_empty}
 }  \label{alg::wait_for_task::shared_empty_end}

 \If{stop == \KwTrue} {                 \label{alg::wait_for_task::stop_beg}
   $notifier[d_w]$.cancel\_wait($w$)\;
   \ForEach{domain $d \in D$} {
     $notifier[d]$.notify\_all()\;
   }
   AtomDec($thieves[d_w]$)\;
   \Return \KwFalse\;
 }                                      \label{alg::wait_for_task::stop_end}

 \If{AtomDec(thieves[$d_w$]) == 0} {  \label{alg::wait_for_task::last_thief_beg}
   \If{$actives[d_w] > 0$} {  \label{alg::wait_for_task::last_thief_check_actives_beg}
     $notifier[d_w]$.cancel\_wait($w$)\;
     \KwGoTo Line 1\;
   } \label{alg::wait_for_task::last_thief_check_actives_end}
   \ForEach{worker $x \in W$} {  \label{alg::wait_for_task::last_thief_check_other_workers_beg}
     \If{x.task\_queue[$d_w$].empty() $\neq$ \KwTrue} {
       $notifier[d_w]$.cancel\_wait($w$)\;
       \KwGoTo Line 1\;
     }
   } \label{alg::wait_for_task::last_thief_check_other_workers_end}
 }  \label{alg::wait_for_task::last_thief_end}
 $notifier[d_w]$.commit\_wait($w$)\; \label{alg::wait_for_task::commit_wait}
 \Return \KwTrue\; \label{alg::wait_for_task::return}
 \caption{wait\_for\_task($w$, $t$)}
 \label{alg::wait_for_task}
\end{algorithm}

When a worker completes all tasks in its local queue,
it proceeds to \texttt{wait\_for\_task}
(line \ref{alg::worker_loop::wait_for_task} in Algorithm \ref{alg::worker_loop}),
as shown in Algorithm \ref{alg::wait_for_task}.
At first, the worker enters \texttt{explore\_task} to make steal attempts 
(line \ref{alg::wait_for_task::explore_task}).
When the worker steals a task and
it is the last thief,
it notifies a worker of the same domain to take over its thief role
and returns to an active worker 
(lines \ref{alg::wait_for_task::thief_to_active_beg}:\ref{alg::wait_for_task::thief_to_active_end}).
Otherwise, 
the worker becomes a \textit{sleep candidate}.
However, we must avoid underutilized parallelism, 
since new tasks may come at the time we put a worker to sleep.
We use 2PC to adapt the number
of active workers to available task parallelism 
(lines \ref{alg::wait_for_task::prepare_wait}:\ref{alg::wait_for_task::commit_wait}).
The predicate of our 2PC is 
\textit{at least one task queue, both local and shared, in the worker's domain is nonempty}.
At line \ref{alg::wait_for_task::thief_to_active_end},
the worker has drained out its local queue and devoted much effort to stealing tasks.
Other task queues in the same domain are most likely to be empty.
We put this worker to a sleep candidate by 
submitting a wait request (line \ref{alg::wait_for_task::prepare_wait}).
From now on, \textit{all the notifications from other workers
will be visible to at least one worker, including this worker}.
That is, if another worker call \texttt{notify} at this moment,
the 2PC guarantees one worker within the scope of 
lines \ref{alg::wait_for_task::prepare_wait}:\ref{alg::wait_for_task::commit_wait}
will be notified (i.e., line \ref{alg::wait_for_task::return}).
Then, we inspect our predicate by examining the shared task queue again 
(lines \ref{alg::wait_for_task::shared_empty_beg}:\ref{alg::wait_for_task::shared_empty_end}),
since external threads might have inserted tasks at the same time we call \texttt{prepare\_wait}.
If the shared queue is nonempty (line \ref{alg::wait_for_task::shared_empty_beg}),  
the worker cancels the wait request and makes an immediate steal attempt
at the queue (lines \ref{alg::wait_for_task::cancel_wait}:\ref{alg::wait_for_task::steal_after_cancel_wait});
if the steal succeeds and it is the last thief,
the worker goes active and notifies a worker 
(lines \ref{alg::wait_for_task::shared_empty_last_thief_beg}:\ref{alg::wait_for_task::shared_empty_last_thief_end}),
or otherwise enters the steal loop again (line \ref{alg::wait_for_task::goto_steal_from_shared_empty}).
If the shared queue is empty (line \ref{alg::wait_for_task::shared_empty_end}),
the worker checks whether the scheduler received a stop signal from the executor
due to exception or task cancellation, 
and notifies all workers to leave 
(lines \ref{alg::wait_for_task::stop_beg}:\ref{alg::wait_for_task::stop_end}).
Now, the worker is almost ready to sleep except if it is the last thief and:
(1) an active worker in its domain exists 
(lines \ref{alg::wait_for_task::last_thief_check_actives_beg}:\ref{alg::wait_for_task::last_thief_check_actives_end}) or
(2) at least one task queue of the same domain from other workers is nonempty 
(lines \ref{alg::wait_for_task::last_thief_check_other_workers_beg}:\ref{alg::wait_for_task::last_thief_check_other_workers_end}).
The two conditions may happen because a task can spawn tasks of different domains
and trigger the scheduler to notify the corresponding domain workers.
Our 2PC guarantees the two conditions 
synchronize with lines \ref{alg::exploit_task::notify_beg}:\ref{alg::exploit_task::notify_end} in Algorithm \ref{alg::exploit_task}
and lines \ref{alg::submit_task::notify_beg}:\ref{alg::submit_task::notify_end} 
in Algorithm \ref{alg::submit_task},
and vice versa,
preventing the problem of undetected task parallelism.
Passing all the above conditions,
the worker commits to wait on our predicate (line \ref{alg::wait_for_task::commit_wait}).

\begin{algorithm}[h]
 \SetKw{KwGoTo}{goto}
 \SetKw{KwBreak}{break}
 \SetKw{KwAnd}{and}
 \SetKw{KwTrue}{true}
 \SetKw{KwFalse}{false}
 \SetKw{KwNIL}{NIL}
 \SetKwInput{KwPerWorkerGlobal}{Per-worker global}
 \KwIn{$w$: a worker (a thief in domain $d_w$)}
 \KwPerWorkerGlobal{$t$: a task (initialized to \KwNIL)}

 $steals$ $\leftarrow$ $0$\;

 \While{t != \KwNIL \KwAnd ++$steals \leq MAX\_STEAL$} {
   yield()\;
   $t \leftarrow$ steal\_task\_from\_random\_victim($d_w$)\;
 }

 \caption{explore\_task($w$, $t$)}
 \label{alg::explore_task}
\end{algorithm}

\begin{algorithm}[h]
 \KwIn{$g$: an HTDG to execute}
 \SetKw{KwAnd}{and}
 \SetKw{KwTrue}{true}
 \SetKw{KwFalse}{false}

 \ForEach{t $\in$ g.source\_tasks} {
   scoped\_lock $lock$($queue\_mutex$)\;
   $d_t \leftarrow t.domain$\;
   $task\_queue[d_t]$.push($t$)\;    \label{alg::submit_graph::push_task}
   $notifier[d_t]$.notify\_one()\;   \label{alg::submit_graph::notify}
 }

 \caption{submit\_graph($g$)}
 \label{alg::submit_graph}
\end{algorithm}

Algorithm \ref{alg::explore_task} implements
\texttt{explore\_task}, which resembles the normal work-stealing loop~\cite{ABP}.
At each iteration, the worker (thief) tries to steal a task
from a randomly selected victim, including the shared task queue,
in the same domain.
We use a parameter $MAX\_STEALS$ to control the number of iterations.
In our experiments, setting $MAX\_STEAL$ to ten times the number of all workers 
is sufficient enough for most applications.
Up to this time,
we have discussed the core work-stealing algorithm.
To submit an HTDG for execution,
we call \texttt{submit\_graph}, shown in Algorithm \ref{alg::submit_graph}.
The caller thread inserts all tasks of zero dependencies
(both strong and weak dependencies) to the shared task queues 
and notifies a worker of the corresponding domain 
(lines \ref{alg::submit_graph::push_task}:\ref{alg::submit_graph::notify}). 
Shared task queues may be accessed
by multiple callers
and are thus protected under a lock pertaining to the scheduler.
Our 2PC guarantees 
lines \ref{alg::submit_graph::push_task}:\ref{alg::submit_graph::notify} 
synchronizes with lines
\ref{alg::wait_for_task::shared_empty_beg}:\ref{alg::wait_for_task::shared_empty_end}
of Algorithm \ref{alg::wait_for_task} and vice versa,
preventing undetected parallelism in which all workers are sleeping.

\section{Analysis}

To justify the efficiency of our scheduling algorithm,
we draw the following theorems and give their proof sketches.

\begin{lemma}
\label{lemma::invariant}
For each domain, when an active worker (i.e., running a task) exists, 
at least one another worker is making steal attempts
unless all workers are active.
\end{lemma}

\begin{proof}
We prove Lemma \ref{lemma::invariant} by contradiction.
Assuming there are no workers making steal attempts when an active worker exists,
this means an active worker (line \ref{alg::exploit_task::notify_beg} in 
Algorithm \ref{alg::exploit_task}) fails to notify one worker if no thieves exist.
There are only two scenarios for this to happen:
(1) all workers are active;
(2) a non-active worker misses the notification before entering the 2PC guard
(line \ref{alg::wait_for_task::prepare_wait} in Algorithm \ref{alg::wait_for_task}).
The first scenario is not possible as it has been excluded by the lemma. 
If the second scenario is true,
the non-active worker must not be the last thief (contradiction) or it will notify another
worker through line \ref{alg::wait_for_task::thief_to_active_beg}
in Algorithm \ref{alg::wait_for_task}.
The proof holds for other domains as our scheduler design is symmetric.
\end{proof}

\begin{theorem}
\label{theorem::correctness}
Our work-stealing algorithm can correctly complete the execution of an HTDG.
\end{theorem}

\begin{proof} 
There are two places where a new task is submitted,
line \ref{alg::submit_graph::push_task} in Algorithm \ref{alg::submit_graph} and
line \ref{alg::submit_task::push_task} in Algorithm \ref{alg::submit_task}.
In the first place,
where a task is pushed to the shared task queue by an external thread,
the notification (line \ref{alg::submit_graph::notify} in Algorithm \ref{alg::submit_graph})
is visible to a worker in the same domain of the task for two situations:
(1) if a worker has prepared or committed to wait 
(lines \ref{alg::wait_for_task::prepare_wait}:\ref{alg::wait_for_task::commit_wait} in Algorithm \ref{alg::wait_for_task}), 
it will be notified;
(2) otherwise, at least one worker will eventually 
go through lines \ref{alg::wait_for_task::prepare_wait}:\ref{alg::wait_for_task::shared_empty_end} in Algorithm \ref{alg::wait_for_task} to steal the task.
In the second place,
where the task is pushed to the corresponding local task queue of that worker,
at least one worker will execute it in either situation:
(1) if the task is in the same domain of the worker, 
the work itself may execute the task in the subsequent \texttt{exploit\_task},
or a thief steals the task through \texttt{explore\_task};
(2) if the worker has a different domain from the task 
(line \ref{alg::submit_task::domain_check} in Algorithm \ref{alg::submit_task}), 
the correctness can be proved by contradiction.
Assuming this task is undetected, which means either the worker
did not notify a corresponding domain worker to run the task 
(false at the condition of line \ref{alg::submit_task::notify_beg} in Algorithm \ref{alg::submit_task})
or notified one worker (line \ref{alg::submit_task::notify} 
in Algorithm \ref{alg::submit_task}) but none have come back.
In the former case, we know at least one worker is active or stealing,
which will eventually go through line
\ref{alg::wait_for_task::last_thief_beg}:\ref{alg::wait_for_task::last_thief_end} 
of Algorithm \ref{alg::wait_for_task} 
to steal this task.
Similarly, the latter case is not possible under our 2PC,
as it contradicts the guarding scan in lines 
\ref{alg::wait_for_task::prepare_wait}:\ref{alg::wait_for_task::commit_wait} 
of Algorithm \ref{alg::wait_for_task}.
\end{proof}

\begin{theorem}
\label{theorem::underscription}
Our work-stealing algorithm does not under-subscribe thread resources during 
the execution of an HTDG.
\end{theorem}

\begin{proof} 
Theorem \ref{theorem::underscription} is a byproduct of Lemma \ref{lemma::invariant} 
and Theorem \ref{theorem::correctness}.
Theorem \ref{theorem::correctness} proves that our scheduler never
has task leak (i.e., undetected task parallelism).
During the execution of an HTDG,
whenever the number of tasks is larger than the present number of workers,
Lemma \ref{lemma::invariant} guarantees one worker is making steal attempts,
unless all workers are active.
The 2PC guard (lines \ref{alg::wait_for_task::last_thief_check_other_workers_beg}:\ref{alg::wait_for_task::last_thief_check_other_workers_end} 
in Algorithm \ref{alg::wait_for_task})
ensures that worker will successfully steal a task and become
an active worker (unless no more tasks), 
which in turn wakes up another worker if that worker is the last thief.
As a consequence, the number of workers will catch up on the number of tasks 
one after one to avoid under-subscribed thread resources.
\end{proof}

\begin{theorem}
\label{theorem::wasteful_steals}
At any moment during the execution of an HTDG,
the number of wasteful steals is bounded by
$\mathcal{O}(MAX\_STEALS\times(|W| + |D|\times(E/e_s)))$,
where $W$ is the worker set, $D$ is the domain set, $E$ is the maximum execution
time of any task, and $e_s$ is the execution time of Algorithm \ref{alg::explore_task}.
\end{theorem}

\begin{proof}
We give a direct proof for Theorem \ref{theorem::wasteful_steals}
using the following notations:
$D$ denotes the domain set,
$d$ denotes a domain (e.g., CPU, GPU),
$W$ denotes the entire worker set,
$W_d$ denotes the worker set in domain $d$,
$w_d$ denotes a worker in domain $d$ (i.e., $w_d \in W_d$),
$e_s$ denotes the time to complete one round of steal attempts
(i.e., Algorithm \ref{alg::explore_task}),
$e_d$ denotes the maximum execution time of any task in domain $d$, and 
$E$ denotes the maximum execution time of any task in the given HTDG.

At any time point,
the worst case happens at the following scenario:
\textit{for each domain $d$ only one worker $w_d$ is actively
running one task while all the other workers are making unsuccessful steal attempts}.
Due to Lemma \ref{lemma::invariant} and 
lines \ref{alg::wait_for_task::last_thief_beg}:\ref{alg::wait_for_task::last_thief_end}
in Algorithm \ref{alg::wait_for_task},
only one thief $w'_d$ will eventually remain in the loop, and the other 
$|W_d| - 2$ thieves will go sleep after one round of 
unsuccessful steal attempts (line \ref{alg::wait_for_task::explore_task}
in Algorithm \ref{alg::wait_for_task}) which ends up with
$MAX\_STEALS\times(|W_d| - 2)$ wasteful steals.
For the only one thief $w'_d$, it keeps failing in steal attempts
until the task running by the only active worker $w_d$ finishes,
and then both go sleep. 
This results in another $MAX\_STEALS\times(e_d/e_s) + MAX\_STEALS$ wasteful steals;
the second terms comes from the active worker because
it needs another round of steal attempts (line \ref{alg::wait_for_task::explore_task}
in Algorithm \ref{alg::wait_for_task}) before going to sleep.
Consequently, the number of wasteful steals across all domains
is bounded as follows:

\begin{equation} \label{eq1}
\begin{split}
 & \sum_{d \in D} MAX\_STEALS\times(|W_d| - 2 + (e_d/e_s) + 1) \\
  & \leq MAX\_STEALS\times\sum_{d \in D} (|W_d| + e_d/e_s) \\
  & \leq MAX\_STEALS\times\sum_{d \in D} (|W_d| + E/e_s) \\
  & = \mathcal{O}(MAX\_STEALS\times(|W| + |D|\times(E/e_s)))
\end{split}
\end{equation}

We do not derive the bound over the execution of an HTDG but the worst-case 
number of wasteful steals at any time point,
because the presence of control flow can lead to non-deterministic execution time
that requires a further assumption of task distribution.
\end{proof}

\section{Experimental Results}

We evaluate the performance of Taskflow on two fronts: 
micro-benchmarks and two realistic workloads, VLSI incremental timing analysis and 
machine learning.
We use micro-benchmarks to analyze the tasking performance of Taskflow
without much bias of application algorithms.
We will show that the performance benefits of Taskflow 
observed in micro-benchmarks
become significant in real workloads.
We will study the performance across runtime, energy efficiency, and throughput.
All experiments ran on a Ubuntu Linux 5.0.0-21-generic x86 64-bit 
machine with 40 Intel Xeon CPU cores at 2.00 GHz,
4 GeForce RTX 2080 GPUs, and 256 GB RAM.
We compiled all programs using Nvidia CUDA v11
on a host compiler of clang++ v10
with C++17 standard \texttt{-std=c++17} 
and optimization flag \texttt{-O2} enabled.
We do not observe significant difference between \texttt{-O2} and \texttt{-O3}
in our experiments.
Each run of $N$ CPU cores and $M$ GPUs corresponds
to $N$ CPU and $M$ GPU worker threads.
All data is an average of 20 runs.

\subsection{Baseline}

Give a large number of TGCSs,
it is impossible to compare Taskflow with all of them.
Each of the existing systems has its pros and cons
and dominates certain applications.
%
%
We consider oneTBB~\cite{TBB}, StarPU~\cite{StarPU}, HPX~\cite{HPX}, 
and OpenMP~\cite{OpenMP} each representing
a particular paradigm that has gained some successful user experiences in CAD
due to performance~\cite{Lu_18_01}.
oneTBB (2021.1 release) 
is an industrial-strength parallel programming system 
under Intel oneAPI~\cite{TBB}.
We consider its FlowGraph library and encapsulate each GPU task in a CPU function.
At the time of this writing, FlowGraph does not have dedicated work stealing for HTDGs.
StarPU (version 1.3) is a CPU-GPU task programming system widely used
in the scientific computing community~\cite{StarPU}.
It provides a C-based syntax for writing HTDGs
on top of a work-stealing runtime highly optimized for CPUs and GPUs.
HPX (version 1.4) is a C++ standard library for concurrency and parallelism~\cite{HPX}.
It supports implicit task graph programming through
aggregating \textit{future} objects in a dataflow API.
OpenMP (version 4.5 in clang toolchains) 
is a directive-based programming framework for handling loop parallelism~\cite{OpenMP}.
It supports static graph encoding 
using task dependency clauses.

%

To measure the expressiveness and programmability of Taskflow,
we hire five PhD-level C++ programmers outside our research group
to implement our experiments.
We educate them the essential knowledge about Taskflow and baseline TGCSs
and provide them all algorithm blocks such that they can focus on 
programming HTDGs.
For each implementation,
we record the lines of code (LOC), the number of tokens, 
cyclomatic complexity (measured by~\cite{SLOCCount}), time to finish,
and the percentage of time spent on debugging.
We average these quantities over five programmers until they 
obtain the correct result.
This measurement may be subjective but it highlights the programming productivity
and turnaround time of each TGCSs from a real user's perspective.

\subsection{Micro-benchmarks}


We randomly generate a set of DAGs (i.e., HTDGs) with equal distribution of CPU and GPU tasks.
Each task performs a SAXPY operation over 1K elements.
For fair purpose, we implemented CUDA Graph~\cite{CUDAGraph} for all baselines;
each GPU task is a CUDA graph of three GPU operations, H2D copy, kernel, and H2D copy,
in this order of dependencies.
Table \ref{tab::programming_effort_on_micro_benchmark} summarizes the 
programming effort of each method.
Taskflow requires the least amount of 
lines of code (LOC) and written tokens.
The cyclomatic complexity of Taskflow measured at a single function and
across the whole program is also the smallest.
The development time of Taskflow-based implementation is much more productive
than the others.
For this simple graph,
Taskflow and oneTBB are very easy for our programmers to implement,
whereas we found they spent a large amount of time on 
debugging task graph parallelism with StarPU, HPX, and OpenMP.

\begin{table}[!h]
\centering
\caption{Programming Effort on Micro-benchmark}
\vspace{-2mm}
\label{tab::programming_effort_on_micro_benchmark}
\resizebox{\columnwidth}{!}{%
\begin{tabular}{|c|c|c|c|c|c|c|}
\hline
\textbf{Method} & \textbf{LOC} & \textbf{\#Tokens} & \textbf{CC} & \textbf{WCC} & \textbf{Dev} & \textbf{Bug} \\
\hline                                                        
Taskflow & 69  & 650  & 6  & 8   & 14 & 1\%  \\
oneTBB   & 182 & 1854 & 8  & 15  & 25 & 6\%  \\
StarPU   & 253 & 2216 & 8  & 21  & 47 & 19\%  \\
HPX      & 255 & 2264 & 10 & 24  & 41 & 33\%  \\
OpenMP   & 182 & 1896 & 13 & 19  & 57 & 49\%  \\
\hline
\end{tabular}
}
\begin{tablenotes}[flushleft]
\footnotesize
\item [1]\textbf{CC}: maximum cyclomatic complexity in a single function
\item [2]\textbf{WCC}: weighted cyclomatic complexity of the program
\item [3]\textbf{Dev}: minutes to complete the implementation
\item [4]\textbf{Bug}: time spent on debugging as opposed to coding task graphs
\end{tablenotes}
\end{table}

Next, we study the overhead of task graph parallelism
among Taskflow, oneTBB, and StarPU.
As shown in Table \ref{tab::comparison_of_task_graph_overhead},
the static size of a task, compiled on our platform,
is 272, 136, and 1472 bytes for Taskflow, oneTBB, and StarPU, respectively.
We do not report the data of HPX and OpenMP 
because they do not support explicit task graph construction at the functional level.
The time it takes for Taskflow 
to create a task and add a dependency is also faster than oneTBB and StarPU.
We amortize the time across 1M operations because all systems support pooled memory
to recycle tasks.
We found StarPU has significant overhead in creating HTDGs.
The overhead always occupies 5-10\% of the total execution time
regardless of the HTDG size.

\begin{table}[!h]
\centering
\caption{Overhead of Task Graph Creation}
\vspace{-2mm}
\label{tab::comparison_of_task_graph_overhead}
\resizebox{\columnwidth}{!}{%
\begin{tabular}{|c|c|c|c|c|c|c|}
\hline
\textbf{Method} & \textbf{$S_{task}$} & \textbf{$T_{task}$} & \textbf{$T_{edge}$} & \textbf{$\rho_{<10}$} & \textbf{$\rho_{<5}$} & \textbf{$\rho_{<1}$} \\
\hline                                                        
Taskflow & 272 & 61 ns  & 14 ns & 550  & 2550 & 35050  \\
oneTBB       & 136 & 99 ns  & 54 ns & 1225 & 2750 & 40050  \\
StarPU       & 1472 & 259 ns & 384 ns & 7550 & - & - \\
\hline
\end{tabular}
}
\begin{tablenotes}[flushleft]
\item [1]\textbf{$S_{task}$}: static size per task in bytes
\item [2]\textbf{$T_{task}/T_{edge}$}: amortized time to create a task/dependency
\item [3]\textbf{$\rho_v$}: graph size where its creation overhead is below $v$\%
\end{tablenotes}
\end{table}

Figure \ref{fig::micro_benchmark_overall_performance} shows the overall performance
comparison between Taskflow and the baseline at different HTDG sizes.
In terms of runtime (top left of Figure \ref{fig::micro_benchmark_overall_performance}), Taskflow outperforms others across most data points.
We complete the largest HTDG by 
1.37$\times$,
1.44$\times$,
1,53$\times$,
and 
1.40$\times$ faster than oneTBB, StarPU, HPX, and OpenMP, respectively.
The memory footprint (top right of Figure \ref{fig::micro_benchmark_overall_performance}) 
of Taskflow is close to oneTBB and OpenMP.
HPX has higher memory because it relies on aggregated futures to
describe task dependencies at the cost of shared states.
Likewise, StarPU does not offer a closure-based interface and thus
requires a flat layout (i.e., codelet) to describe tasks.
We use the Linux \texttt{perf} tool to measure the power consumption of all 
cores plus LLC~\cite{PerfStat}.
The total joules (bottom left of Figure \ref{fig::micro_benchmark_overall_performance}) consumed by Taskflow is consistently smaller than the others,
due to our adaptive worker management.
In terms of power (bottom right of Figure \ref{fig::micro_benchmark_overall_performance}), Taskflow, oneTBB, and OpenMP are more power-efficient than HPX and StarPU.
The difference between Taskflow and StarPU continues to increase as
we enlarge the HTDG size.

\begin{figure}[!h]
  \centering
  \pgfplotsset{
    title style={font=\LARGE},
    label style={font=\large},
  }
  \begin{tikzpicture}[scale=0.49]
    \begin{axis}[
      title=Runtime,
      ylabel=Runtime (ms),
      xlabel=Graph Size ($|V|+|E|$),
      legend pos=north west,
    ]
    \addplot+ table[x=size,y=tf,col sep=space]{Fig/micro_benchmark/runtime-server-40C4G.txt};
    \addplot+ table[x=size,y=oneTBB,col sep=space]{Fig/micro_benchmark/runtime-server-40C4G.txt};
    \addplot+ [mark=diamond,color=cyan]table[x=size,y=starpu,col sep=space]{Fig/micro_benchmark/runtime-server-40C4G.txt};
    \addplot+ table[x=size,y=hpx,col sep=space]{Fig/micro_benchmark/runtime-server-40C4G.txt};
    \addplot+ [mark=o,color=brown] table[x=size,y=omp,col sep=space]{Fig/micro_benchmark/runtime-server-40C4G.txt};
    \legend{Taskflow, oneTBB, StarPU, HPX, OpenMP}
    \end{axis}
  \end{tikzpicture}
  \begin{tikzpicture}[scale=0.49]
    \begin{axis}[
      title=Memory,
      ylabel=Maximum RSS (MB),
      xlabel=Graph Size ($|V|+|E|$),
      legend pos=north west,
      ymax=650,
      y filter/.code={\pgfmathparse{#1/1000}\pgfmathresult}
    ]
    \addplot+ table[x=size,y=tf,col sep=space]{Fig/micro_benchmark/memory-server-40C4G.txt};
    \addplot+ table[x=size,y=oneTBB,col sep=space]{Fig/micro_benchmark/memory-server-40C4G.txt};
    \addplot+ [mark=diamond,color=cyan] table[x=size,y=starpu,col sep=space]{Fig/micro_benchmark/memory-server-40C4G.txt};
    \addplot+ table[x=size,y=hpx,col sep=space]{Fig/micro_benchmark/memory-server-40C4G.txt};
    \addplot+ [mark=o,color=brown] table[x=size,y=omp,col sep=space]{Fig/micro_benchmark/memory-server-40C4G.txt};
    \legend{Taskflow, oneTBB, StarPU, HPX, OpenMP}
    \end{axis}
  \end{tikzpicture}
  \begin{tikzpicture}[scale=0.49]
    \begin{axis}[
      title=Energy,
      ylabel=Total Joules (J),
      xlabel=Graph Size ($|V|+|E|$),
      legend pos=north west,
      ymax=1100
    ]
    \addplot+ table[x=size,y=tf-j,col sep=space]{Fig/micro_benchmark/energy-server-40C4G.txt};
    \addplot+ table[x=size,y=oneTBB-j,col sep=space]{Fig/micro_benchmark/energy-server-40C4G.txt};
    \addplot+ [mark=diamond,color=cyan] table[x=size,y=starpu-j,col sep=space]{Fig/micro_benchmark/energy-server-40C4G.txt};
    \addplot+ table[x=size,y=hpx-j,col sep=space]{Fig/micro_benchmark/energy-server-40C4G.txt};
    \addplot+ [mark=o,color=brown] table[x=size,y=omp-j,col sep=space]{Fig/micro_benchmark/energy-server-40C4G.txt};
    \legend{Taskflow, oneTBB, StarPU, HPX, OpenMP}
    \end{axis}
  \end{tikzpicture}
  \begin{tikzpicture}[scale=0.49]
    \begin{axis}[
      title=Power (all cores + LLC),
      ylabel=Average Power (W),
      xlabel=Graph Size ($|V|+|E|$),
      legend style={at={(0.53,0.58)},anchor=west}
    ]
    \addplot+ table[x=size,y=tf-p,col sep=space]{Fig/micro_benchmark/energy-server-40C4G.txt};
    \addplot+ table[x=size,y=oneTBB-p,col sep=space]{Fig/micro_benchmark/energy-server-40C4G.txt};
    \addplot+ [mark=diamond,color=cyan] table[x=size,y=starpu-p,col sep=space]{Fig/micro_benchmark/energy-server-40C4G.txt};
    \addplot+ table[x=size,y=hpx-p,col sep=space]{Fig/micro_benchmark/energy-server-40C4G.txt};
    \addplot+ [mark=o,color=brown] table[x=size,y=omp-p,col sep=space]{Fig/micro_benchmark/energy-server-40C4G.txt};
    \legend{Taskflow, oneTBB, StarPU, HPX, OpenMP}
    \end{axis}
  \end{tikzpicture}
  \caption{Overall system performance at different problem sizes using 40 CPUs and 4 GPUs.}
  \label{fig::micro_benchmark_overall_performance}
\end{figure}
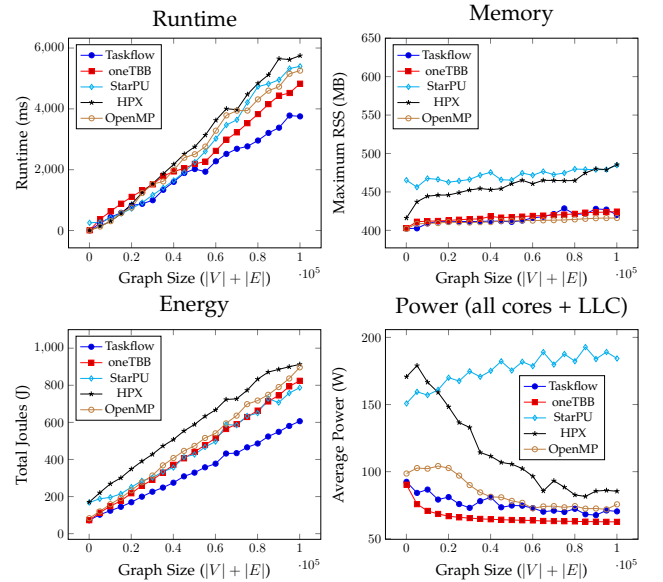

\begin{figure}[h]
  \centering
  \pgfplotsset{
    title style={font=\LARGE},
    label style={font=\large},
  }
  \begin{tikzpicture}[scale=0.50]
    \begin{axis}
    [
      xtick={1,2,3,4,5},
      xticklabels={Taskflow, oneTBB, StarPU, HPX, OpenMP},
      xmajorgrids=true,
      x tick label style={rotate=90,anchor=east},
      boxplot/draw direction=y,
      title=Runtime (5K tasks),
      ylabel=Runtime (ms),
    ]
    \addplot+[
    boxplot prepared={
      lower whisker=778.211,
      lower quartile=786.76,
      median=793.33,
      upper quartile=802.4645,
      upper whisker=812.616
    },
    ] coordinates {};
    \addplot+[ 
    boxplot prepared={
      lower whisker=719.381,
      lower quartile=728.9465,
      median=770.8885,
      upper quartile=895.3375,
      upper whisker=993.14
    },
    ] coordinates {};
    \addplot+ [color=cyan] [ 
    boxplot prepared={
      lower whisker=1237.23,
      lower quartile=1248.045,
      median=1246.12,
      upper quartile=1248.045,
      upper whisker=1291.98
    },
    ] coordinates {};
    \addplot+[ 
    boxplot prepared={
      lower whisker=985.465,
      lower quartile=1046.4695,
      median=1155.41,
      upper quartile=1209.3275,
      upper whisker=1296.44
    },
    ] coordinates {};
    \addplot+ [color=brown][
    boxplot prepared={
      lower whisker=722.82,
      lower quartile=759.8735,
      median=806.294,
      upper quartile=901.709,
      upper whisker=1099.51
    },
    ] coordinates {};
  \end{axis}
  \end{tikzpicture}
  \begin{tikzpicture}[scale=0.50]
    \begin{axis}
    [
      xtick={1,2,3,4,5},
      xticklabels={Taskflow, oneTBB, StarPU, HPX, OpenMP},
      xmajorgrids=true,
      x tick label style={rotate=90,anchor=east},
      boxplot/draw direction=y,
      title=Runtime (20K tasks),
      ylabel=Runtime (ms),
    ]
    \addplot+[
    boxplot prepared={
      lower whisker=4411.54,
      lower quartile=4432.5925,
      median=4491.37,
      upper quartile=4500.8025,
      upper whisker=4524.31
    },
    ] coordinates {};
    \addplot+[ 
    boxplot prepared={
      lower whisker=5078.38,
      lower quartile=5114.3125,
      median=5235.825,
      upper quartile=5354.9975,
      upper whisker=5398.15
    },
    ] coordinates {};
    \addplot+ [color=cyan] [ 
    boxplot prepared={
      lower whisker=5247.38,
      lower quartile=5397.8575,
      median=5533.905,
      upper quartile=5727.895,
      upper whisker=5900.47
    },
    ] coordinates {};
    \addplot+[ 
    boxplot prepared={
      lower whisker=5581.73,
      lower quartile=5672.615,
      median=5678.92,
      upper quartile=5882.1675,
      upper whisker=5967.13
    },
    ] coordinates {};
    \addplot+ [color=brown][
    boxplot prepared={
      lower whisker=4883.89,
      lower quartile=5369.155,
      median=5577.67,
      upper quartile=5746.05,
      upper whisker=5820.84
    },
    ] coordinates {};
  \end{axis}
  \end{tikzpicture}
  \caption{Runtime distribution of two task graphs.}
  \label{fig::micro_benchmark_boxplot}
\end{figure}
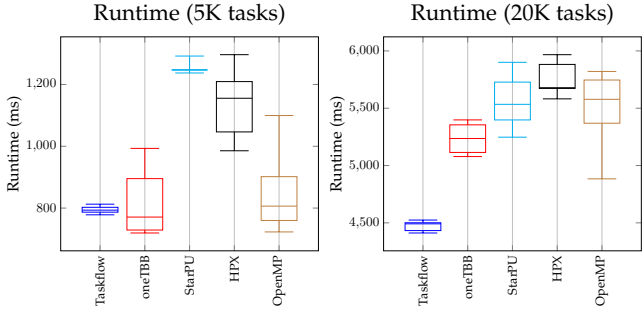

Figure \ref{fig::micro_benchmark_boxplot} displays the runtime distribution
of each method over a hundred runs of two HTDGs, 5K and 20K tasks.
The boxplot shows that the runtime of Taskflow is more consistent than others
and has the smallest variation.
We attribute this result to the design of our scheduler,
which effectively separates task execution into CPU and GPU workers
and dynamically balances cross-domain wasteful steals
with task parallelism.

\begin{figure}[!h]
  \centering
  \pgfplotsset{
    title style={font=\LARGE},
    label style={font=\large},
  }
  \begin{tikzpicture}[scale=0.50]
    \begin{axis}[
      title= Corun (20K tasks),
      ylabel=Throughput,
      xlabel=Number of Coruns,
      legend pos=north west,
      ymin=0,
      ymax=2.5
    ]
    \addplot+ table[x=corun,y=tf,col sep=space]{Fig/micro_benchmark/throughput-server-40C4G.txt};
    \addplot+ table[x=corun,y=oneTBB,col sep=space]{Fig/micro_benchmark/throughput-server-40C4G.txt};
    \addplot+ [mark=diamond,color=cyan] table[x=corun,y=starpu,col sep=space]{Fig/micro_benchmark/throughput-server-40C4G.txt};
    \addplot+ table[x=corun,y=hpx,col sep=space]{Fig/micro_benchmark/throughput-server-40C4G.txt};
    \addplot+ [mark=o,color=brown] table[x=corun,y=omp,col sep=space]{Fig/micro_benchmark/throughput-server-40C4G.txt};
    \legend{Taskflow, oneTBB, StarPU, HPX, OpenMP}
    \end{axis}
  \end{tikzpicture}
  \begin{tikzpicture}[scale=0.50]
    \begin{axis}[
      title=Consumed CPU Resource,
      ylabel=Utilization,
      xlabel=Graph Size ($|V|+|E|$),
      legend style={at={(0.53,0.6)},anchor=west},
    ]
    \addplot+ table[x=size,y=tf,col sep=space]{Fig/micro_benchmark/utilization-server-40C4G.txt};
    \addplot+ table[x=size,y=oneTBB,col sep=space]{Fig/micro_benchmark/utilization-server-40C4G.txt};
    \addplot+ [mark=diamond,color=cyan] table[x=size,y=starpu,col sep=space]{Fig/micro_benchmark/utilization-server-40C4G.txt};
    \addplot+ table[x=size,y=hpx,col sep=space]{Fig/micro_benchmark/utilization-server-40C4G.txt};
    \addplot+ [mark=o,color=brown] table[x=size,y=omp,col sep=space]{Fig/micro_benchmark/utilization-server-40C4G.txt};
    \legend{Taskflow, oneTBB, StarPU, HPX, OpenMP}
    \end{axis}
  \end{tikzpicture}
  \caption{Throughput of corunning task graphs and CPU utilization at different problem sizes under 40 CPUs and 4 GPUs.}
  \label{fig::micro_benchmark_throughput}
\end{figure}
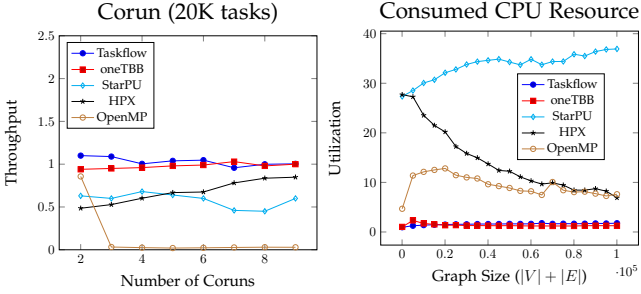

Finally, we compare the throughput of each method on corunning HTDGs.
This experiment emulates a server-like environment
where multiple programs run simultaneously
on the same machine to compete for the same resources.
The effect of worker management propagates to all parallel processes.
We consider up to nine corun processes 
each executing the same HTDG of 20K tasks.
We use the weighted speedup~\cite{BWS} to measure the system throughput.
Figure \ref{fig::micro_benchmark_throughput} compares the throughput 
of each method and relates the result to the CPU utilization.
Both Taskflow and oneTBB produce significantly higher throughput than others.
Our throughput is slightly better than oneTBB by 1--15\% except for seven coruns.
The result can be interpreted by the CPU utilization plot, 
reported by \texttt{perf stat}.
We can see both Taskflow and oneTBB make effective use of CPU resources
to schedule tasks.
However, StarPU keeps workers busy most of the time and has no mechanism
to dynamically control thread resources with task parallelism.

Since both oneTBB and StarPU provides explicit task graph programming models
and work-stealing for dynamic load balancing,
we will focus on comparing Taskflow with oneTBB and StarPU for the next two real
workloads.

\subsection{VLSI Incremental Timing Analysis}

As part of our DARPA project,
we applied Taskflow to solve a VLSI incremental static timing analysis (STA) problem
in an optimization loop.
The goal is to optimize the timing landscape of a circuit design
by iteratively applying \textit{design transforms} (e.g., gate sizing, buffer insertion)
and evaluating the timing improvement until all data paths are passing,
aka \textit{timing closure}.
Achieving timing closure is one of the most time-consuming steps in the 
VLSI design closure flow process because optimization algorithms can call a timer 
millions or even billions of times to incrementally analyze the timing improvement 
of a design transform.
We consider the GPU-accelerated critical path analysis algorithm~\cite{Huang_21_03}
and run it across one thousand incremental iterations
based on the design transforms given by TAU 2015 Contest~\cite{Hu_15_01}.
The data is generated by an industrial tool to evaluate the performance
of an incremental timing algorithm.
Each incremental iteration corresponds to at least one design modifier
followed by a timing report operation to trigger incremental timing update 
of the timer.

\begin{figure}[!h]
  \centering
  \centerline{\includegraphics[width=1.\columnwidth]{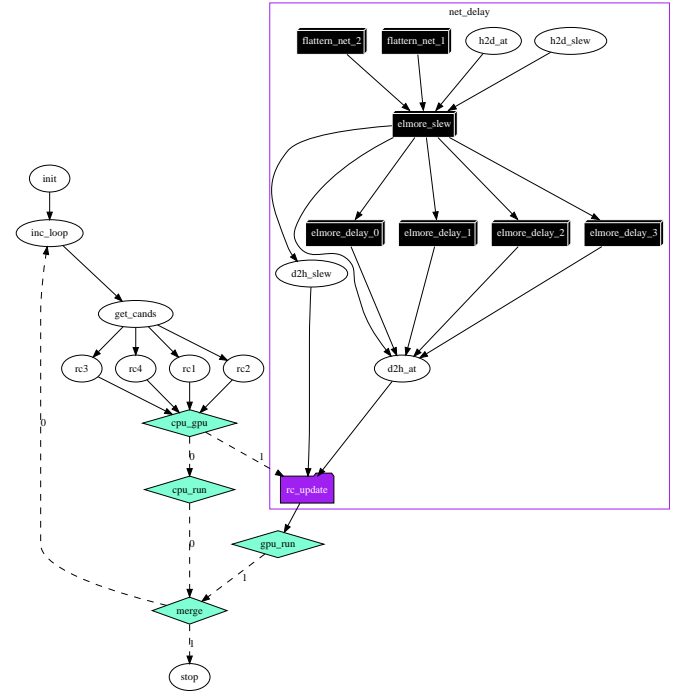}}
  \caption{A partial HTDG of 1 cudaFlow task (purple box), 4 condition tasks (green diamond), and 8 static tasks (other else) for one iteration of timing-driven optimization.}
  \label{fig::sta}
\end{figure}

Figure \ref{fig::sta} shows a partial Taskflow graph of our implementation.
One condition task forms a loop to implement iterative timing updates and
the other three condition tasks branch the execution to either CPU-based
timing update (over 10K tasks) or GPU-based timing update (cudaFlow tasks).
The motivation here is to adapt the timing update to
different incrementalities.
For example, if a design transform introduces only a few hundreds of nodes
to update, there is no need to offload the computation to GPUs
due to insufficient amount of data parallelism.
The cudaFlow task composes over 1K operations to compute large interconnect delays,
which often involves several gigabytes of parasitic data.
Since oneTBB FlowGraph and StarPU do not support control flow,
we unroll their task graphs across fixed-length iterations found in hindsight
to avoid expensive synchronization at each iteration;
the number of concatenated graphs is equal to the number of iterations.

\begin{table}[!h]
\centering
\caption{Programming Effort on VLSI Timing Closure}
\vspace{-2mm}
\label{tab::programming_effort_on_timing_closure}
\resizebox{\columnwidth}{!}{%
\begin{tabular}{|c|c|c|c|c|c|c|}
\hline
\textbf{Method} & \textbf{LOC} & \textbf{\#Tokens} & \textbf{CC} & \textbf{WCC} & \textbf{Dev} & \textbf{Bug} \\
\hline                                                        
Taskflow & 3176 & 5989  & 30  & 67 & 3.9 & 13\%  \\
oneTBB   & 4671 & 8713  & 41  & 92 & 6.1 & 51\%  \\
StarPU   & 5643 & 13952 & 46  & 98 & 4.3 & 38\%  \\
\hline
\end{tabular}
}
\begin{tablenotes}
\footnotesize
\item [1]\textbf{CC}: maximum cyclomatic complexity in a single function
\item [2]\textbf{WCC}: weighted cyclomatic complexity of the program
\item [3]\textbf{Dev}: hours to complete the implementation
\item [4]\textbf{Bug}: time spent on the debugging versus coding task graphs
\end{tablenotes}
\end{table}

Table \ref{tab::programming_effort_on_timing_closure} 
compares the programming effort between Taskflow, oneTBB, 
and StarPU.
In a rough view, the implementation complexity using Taskflow
is much less than that of oneTBB and StarPU.
The amount of time spent on implementing the algorithm is about 3.9 hours
for Taskflow, 6.1 hours for oneTBB, and 4.3 hours for StarPU.
It takes 3--4$\times$ more time to debug oneTBB and StarPU than Taskflow,
mostly on control flow.
Interestingly, 
while StarPU involves more LOC and higher cyclomatic complexity than oneTBB,
our programmers found StarPU easier to write due to its C-styled interface.
Although there is no standard way to conclude the programmability of a library,
we believe our measurement highlights the 
expressiveness of Taskflow and its ease of use from a real user's perspective.

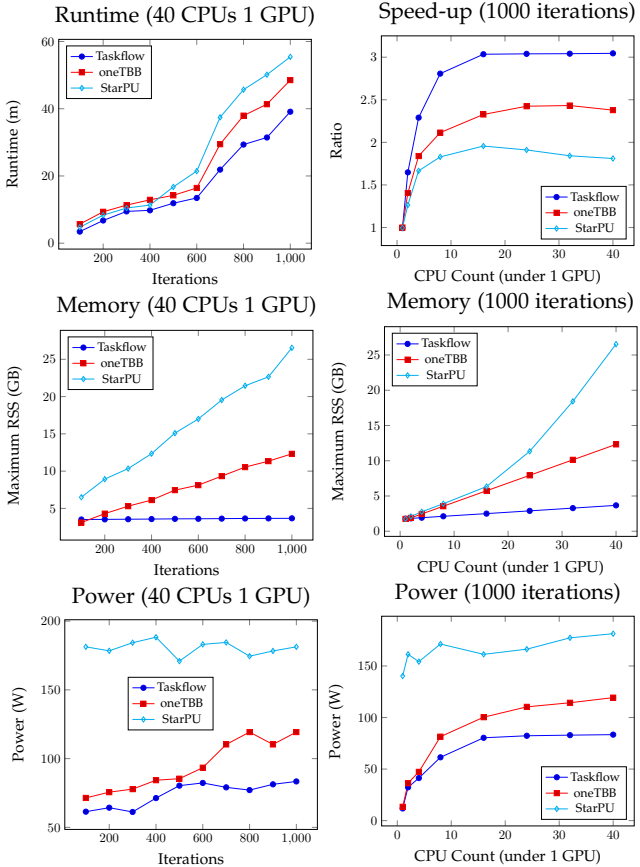
\begin{figure}[!h]
  \centering
  \pgfplotsset{
    title style={font=\LARGE},
    label style={font=\large},
  }
  \begin{tikzpicture}[scale=0.49]
    \begin{axis}[
      title=Runtime (40 CPUs 1 GPU),
      ylabel=Runtime (m),
      xlabel=Iterations,
      legend pos=north west,
    ]
    \addplot+ table[x=iter,y=tf-8C1G,col sep=space]{Fig/placement/adaptec1_problem_scalability.txt};
    \addplot+ table[x=iter,y=oneTBB-8C1G,col sep=space]{Fig/placement/adaptec1_problem_scalability.txt};
    \addplot+ [mark=diamond,color=cyan]table[x=iter,y=starpu-8C1G,col sep=space]{Fig/placement/adaptec1_problem_scalability.txt};
    \legend{Taskflow, oneTBB, StarPU}
    \end{axis}
  \end{tikzpicture}
  \begin{tikzpicture}[scale=0.49]
    \begin{axis}[
      title=Speed-up (1000 iterations),
      ylabel=Ratio,
      xlabel=CPU Count (under 1 GPU),
      legend pos=south east,
      y filter/.code={\pgfmathparse{113/#1}\pgfmathresult}
    ]
    \addplot+ table[x=CPU,y=tf,col sep=space]{Fig/placement/adaptec1_cpu_scalability.txt};
    \addplot+ table[x=CPU,y=oneTBB,col sep=space]{Fig/placement/adaptec1_cpu_scalability.txt};
    \addplot+ [mark=diamond,color=cyan]table[x=CPU,y=starpu,col sep=space]{Fig/placement/adaptec1_cpu_scalability.txt};
    \legend{Taskflow, oneTBB, StarPU}
    \end{axis}
  \end{tikzpicture}
  \begin{tikzpicture}[scale=0.49]
    \begin{axis}[
      title=Memory (40 CPUs 1 GPU),
      ylabel=Maximum RSS (GB),
      xlabel=Iterations,
      legend pos=north west,
    ]
    \addplot+ table[x=iter,y=tf-mem-8C1G,col sep=space]{Fig/placement/adaptec1_problem_scalability.txt};
    \addplot+ table[x=iter,y=oneTBB-mem-8C1G,col sep=space]{Fig/placement/adaptec1_problem_scalability.txt};
    \addplot+ [mark=diamond,color=cyan]table[x=iter,y=starpu-mem-8C1G,col sep=space]{Fig/placement/adaptec1_problem_scalability.txt};
    \legend{Taskflow, oneTBB, StarPU}
    \end{axis}
  \end{tikzpicture}
  \begin{tikzpicture}[scale=0.49]
    \begin{axis}[
      title=Memory (1000 iterations),
      ylabel=Maximum RSS (GB),
      xlabel=CPU Count (under 1 GPU),
      legend pos=north west,
    ]
    \addplot+ table[x=CPU,y=tf-mem,col sep=space]{Fig/placement/adaptec1_cpu_scalability.txt};
    \addplot+ table[x=CPU,y=oneTBB-mem,col sep=space]{Fig/placement/adaptec1_cpu_scalability.txt};
    \addplot+ [mark=diamond,color=cyan]table[x=CPU,y=starpu-mem,col sep=space]{Fig/placement/adaptec1_cpu_scalability.txt};
    \legend{Taskflow, oneTBB, StarPU}
    \end{axis}
  \end{tikzpicture}
  \begin{tikzpicture}[scale=0.49]
    \begin{axis}[
      title=Power (40 CPUs 1 GPU),
      ylabel=Power (W),
      xlabel=Iterations,
      legend style={at={(0.25,0.60)},anchor=west}
    ]
    \addplot+ table[x=iter,y=tf-p,col sep=space]{Fig/placement/adaptec1_problem_scalability.txt};
    \addplot+ table[x=iter,y=oneTBB-p,col sep=space]{Fig/placement/adaptec1_problem_scalability.txt};
    \addplot+ [mark=diamond,color=cyan]table[x=iter,y=starpu-p,col sep=space]{Fig/placement/adaptec1_problem_scalability.txt};
    \legend{Taskflow, oneTBB, StarPU}
    \end{axis}
  \end{tikzpicture}
  \begin{tikzpicture}[scale=0.49]
    \begin{axis}[
      title=Power (1000 iterations),
      ylabel=Power (W),
      xlabel=CPU Count (under 1 GPU),
      legend pos=south east,
    ]
    \addplot+ table[x=CPU,y=tf-p,col sep=space]{Fig/placement/adaptec1_cpu_scalability.txt};
    \addplot+ table[x=CPU,y=oneTBB-p,col sep=space]{Fig/placement/adaptec1_cpu_scalability.txt};
    \addplot+ [mark=diamond,color=cyan]table[x=CPU,y=starpu-p,col sep=space]{Fig/placement/adaptec1_cpu_scalability.txt};
    \legend{Taskflow, oneTBB, StarPU}
    \end{axis}
  \end{tikzpicture}
  \caption{Runtime, memory, and power data of 1000 incremental timing iterations (up to 11K tasks and 17K dependencies per iteration) on a large design of 1.6M gates.}
  \label{fig::adaptec1_scalability}
\end{figure}

The overall performance is shown in Figure \ref{fig::adaptec1_scalability}.
Using 40 CPUs and 1 GPU,
Taskflow is consistently faster than oneTBB and StarPU
across all incremental timing iterations.
The gap continues to enlarge as increasing iteration numbers;
at 100 and 1000 iterations, Taskflow reaches the goal in 3.45 and 39.11 minutes,
whereas oneTBB requires 5.67 and 4.76 minutes and StarPU requires 48.51 and 55.43 minutes, 
respectively.
Note that the gain is significant because a typical timing closure algorithm
can invoke millions to billions of iterations 
that take several hours to finish~\cite{Huang_20_01}.
We observed similar results at other CPU numbers;
in terms of the runtime speed-up over 1 CPU (all finish in 113 minutes),
Taskflow is always faster than oneTBB and StarPU, regardless of the CPU count.
Speed-up of Taskflow saturates at about 16 CPUs ($3\times$), 
primarily due to the inherent irregularity of the algorithm
(see Figure \ref{fig::sta}).
The memory footprint (middle of Figure \ref{fig::adaptec1_scalability}) 
shows the benefit of our conditional tasking.
By reusing condition tasks in the incremental timing loop,
we do not suffer significant memory growth as oneTBB and StarPU.
On a vertical scale,
increasing the number of CPUs bumps up the memory usage of both methods,
but Taskflow consumes much less
because we use only simple atomic operations to control
wasteful steals.
In terms of energy efficiency (bottom of Figure \ref{fig::adaptec1_scalability},
measured on all cores plus LLC using \texttt{power/energy-pkg}~\cite{PerfStat}),
our scheduler is very power-efficient
in completing the timing analysis workload, regardless of iterations and CPU numbers.
Beyond 16 CPUs where performance saturates, Taskflow does not suffer from
increasing power as oneTBB and StarPU, because
our scheduler efficiently balances the number of workers with dynamic
task parallelism.

\begin{figure}[!h]
  \centering
  \pgfplotsset{
    title style={font=\LARGE},
    label style={font=\large},
  }
  \begin{tikzpicture}[scale=0.49]
    \begin{axis}[
      title= Corun (500 iterations),
      ylabel=Throughput,
      xlabel=Number of Coruns,
      legend pos=north west,
      ymin=0,
      ymax=5.5
    ]
    \addplot+ table[x=corun,y=tf-50,col sep=space]{Fig/placement/throughput.txt};
    \addplot+ table[x=corun,y=oneTBB-50,col sep=space]{Fig/placement/throughput.txt};
    \addplot+ [mark=diamond,color=cyan] table[x=corun,y=starpu-50,col sep=space]{Fig/placement/throughput.txt};
    \legend{Taskflow, oneTBB, StarPU}
    \end{axis}
  \end{tikzpicture}
  \begin{tikzpicture}[scale=0.49]
    \begin{axis}[
      title= Corun (1000 Iterations),
      ylabel=Throughput,
      xlabel=Number of Coruns,
      legend pos=south east,
      ymin=0,
    ]
    \addplot+ table[x=corun,y=tf-100,col sep=space]{Fig/placement/throughput.txt};
    \addplot+ table[x=corun,y=oneTBB-100,col sep=space]{Fig/placement/throughput.txt};
    \addplot+ [mark=diamond,color=cyan] table[x=corun,y=starpu-100,col sep=space]{Fig/placement/throughput.txt};
    \legend{Taskflow, oneTBB, StarPU} 
    \end{axis}
  \end{tikzpicture}
  \caption{Throughput of corunning timing analysis  workloads on two iteration numbers using 40 CPUs and 1 GPU.}
  \label{fig::placement_throughput}
\end{figure}
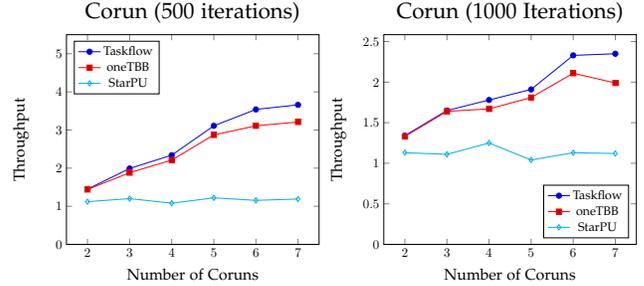

We next compare the throughput of each implementation
by corunning the same program.
Corunning programs is a common strategy for optimization
tools to search for the best parameters.
The effect of worker management propagates to all simultaneous processes.
Thus, the throughput can be a good measurement for the 
inter-operability of a scheduling algorithm.
We corun the same timing analysis program up to seven processes 
that compete for 40 CPUs and 1 GPU.
We use the \textit{weighted speedup} to measure the system throughput,
which is the sum of the individual speedup of each process 
over a baseline execution time~\cite{BWS}.
A throughput of one implies that the corun's throughput is the
same as if the processes were run consecutively.
Figure \ref{fig::placement_throughput} plots the throughput across nine coruns
at two iteration numbers.
Both Taskflow and oneTBB achieve decent throughput greater than one
and are significantly better than StarPU.
We found StarPU keep workers busy most of the time and has no mechanism
to balance the number of workers with dynamically generated task parallelism.
For irregular HTDGs akin to Figure \ref{fig::sta},
worker management is critical for corunning processes.
When task parallelism becomes sparse, especially around 
the decision-making point of an iterative control flow,
our scheduler can adaptively reduce the wasteful steals 
based on the active worker count,
and we offer a stronger bound than oneTBB (Theorem \ref{theorem::wasteful_steals}).
Saved wasteful resources can thus be used by other concurrent programs
to increase the throughput.
%
%

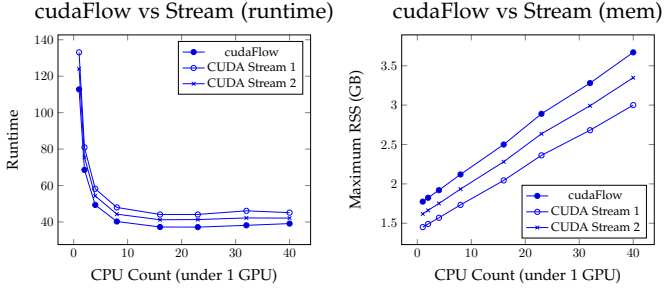
\begin{figure}[!h]
  \centering
  \pgfplotsset{
    title style={font=\LARGE},
    label style={font=\large},
  }
  \begin{tikzpicture}[scale=0.49]
    \begin{axis}[
      title= cudaFlow vs Stream (runtime),
      ylabel=Runtime,
      xlabel=CPU Count (under 1 GPU),
      legend pos=north east,
    ]
    \addplot+ table[x=cpu,y=tf-cpu,col sep=space]{Fig/placement/cudaflow_overhead.txt};
    \addplot+ [mark=o,color=blue] table[x=cpu,y=s1-cpu,col sep=space]{Fig/placement/cudaflow_overhead.txt};
    \addplot+ [mark=x,color=blue] table[x=cpu,y=s2-cpu,col sep=space]{Fig/placement/cudaflow_overhead.txt};
    \legend{cudaFlow, CUDA Stream 1, CUDA Stream 2}
    \end{axis}
  \end{tikzpicture}
  \begin{tikzpicture}[scale=0.49]
    \begin{axis}[
      title= cudaFlow vs Stream (mem),
      xlabel=CPU Count (under 1 GPU),
      ylabel=Maximum RSS (GB),
      legend pos=south east,
    ]
    \addplot+ table[x=cpu,y=tf-mem,col sep=space]{Fig/placement/cudaflow_overhead.txt};
    \addplot+ [mark=o,color=blue]table[x=cpu,y=s1-mem,col sep=space]{Fig/placement/cudaflow_overhead.txt};
    \addplot+ [mark=x,color=blue] table[x=cpu,y=s2-mem,col sep=space]{Fig/placement/cudaflow_overhead.txt};
    \legend{cudaFlow, CUDA Stream 1, CUDA Stream 2}
    \end{axis}
  \end{tikzpicture}
  \caption{Comparison of runtime and memory between cudaFlow (CUDA Graph) and stream-based execution in the VLSI incremental timing analysis workload.}
  \label{fig::sta_cudaflow_overhead}
\end{figure}

Figure \ref{fig::sta_cudaflow_overhead} shows the performance advantage of CUDA Graph
and its cost in handling this large GPU-accelerated timing analysis workloads.
The line \textit{cudaFlow} represents our default implementation 
using explicit CUDA graph construction.
The other two lines represent the implementation of the same GPU task graph but using
stream and event insertions (i.e., non-CUDA Graph).
As partially shown in Figure \ref{fig::sta},
our cudaFlow composes over 1K dependent GPU operations to compute the interconnect 
delays.
For large GPU workloads like this, the benefit of CUDA Graph is clear;
we observed 9--17\% runtime speed-up over stream-based implementations.
The performance improvement mostly comes from reduced kernel call overheads
and graph-level scheduling optimizations by CUDA runtime.
Despite the improved performance, cudaFlow incurs higher memory costs
because CUDA Graph stores all kernel parameters in advance for optimization.
For instance, creating a node in CUDA Graph can take over 300 bytes of opaque data structures.

\subsection{Large Sparse Neural Network Inference}

We applied Taskflow to solve the 
MIT/Amazon Large Sparse Deep Neural Network (LSDNN) Inference Challenge,
a recent effort aimed at new computing methods for sparse AI analytics~\cite{Kepner_19_01}.
Each dataset comprises a sparse matrix of the input data for the network,
1920 layers of neurons stored in sparse matrices, truth categories, and the bias 
values used for the inference.
Preloading the network to the GPU is impossible. Thus, we implement a 
model decomposition-based kernel algorithm inspired by~\cite{Bisson_19_01}
and construct an end-to-end HTDG for the entire inference workload.
Unlike VLSI incremental timing analysis,
this workload is both CPU- and GPU-heavy.
Figure \ref{fig::snig} illustrates a partial HTDG.
We create up to 4 cudaFlows on 4 GPUs.
Each cudaFlow contains more than 2K GPU operations to run 
partitioned matrices in an iterative data dispatching loop
formed by a condition task.
Other CPU tasks evaluate the results with a golden reference.
Since oneTBB FlowGraph and StarPU do not support in-graph control flow,
we unroll their task graph across fixed-length iterations found offline.

\begin{figure}[h]
  \centering
  \centerline{\includegraphics[width=1.\columnwidth]{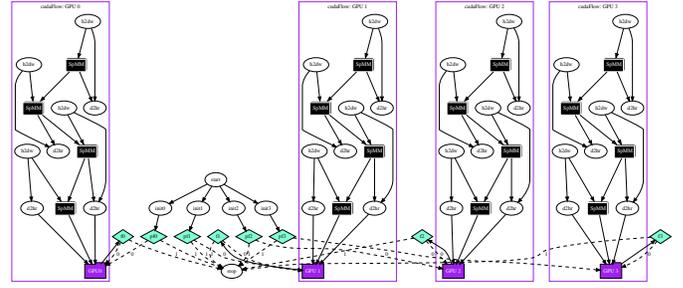}}
  \caption{A partial HTDG of 4 cudaFlows (purple boxes), 8 conditioned cycles (green diamonds),
  and 6 static tasks (other else) for the inference workload.}
  \label{fig::snig}
\end{figure}


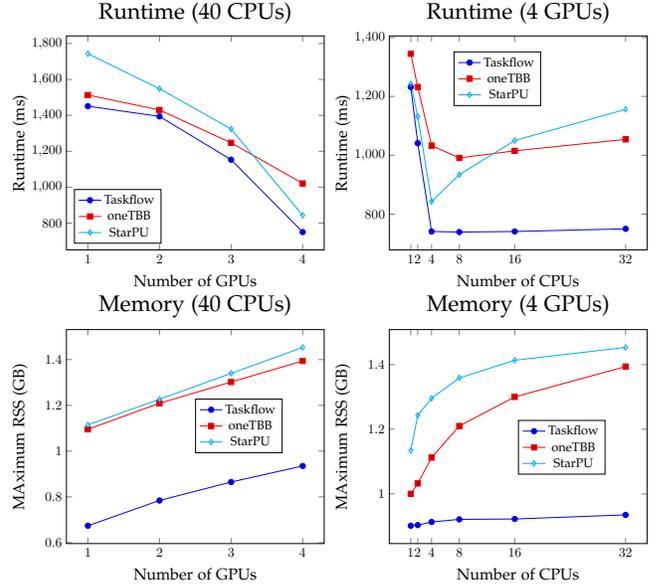
\begin{figure}[!h]
  \centering
  \pgfplotsset{
    title style={font=\LARGE},
    label style={font=\large},
  }
  \begin{tikzpicture}[scale=0.50]
    \begin{axis}[
      title= Runtime (40 CPUs),
      ylabel=Runtime (ms),
      xlabel=Number of GPUs,
      legend pos=south west,
      xtick={1, 2, 3, 4},
    ]
    \addplot+ table[x=gpus,y=tf,col sep=space]{Fig/dnn/scalability-40C.txt};
    \addplot+ table[x=gpus,y=oneTBB,col sep=space]{Fig/dnn/scalability-40C.txt};
    \addplot+ [mark=diamond,color=cyan] table[x=gpus,y=starpu,col sep=space]{Fig/dnn/scalability-40C.txt};
    \legend{Taskflow, oneTBB, StarPU}
    \end{axis}
  \end{tikzpicture}
  \begin{tikzpicture}[scale=0.50]
    \begin{axis}[
      title= Runtime (4 GPUs),
      ylabel=Runtime (ms),
      xlabel=Number of CPUs,
      xtick={1, 2, 4, 8, 16, 32},
      legend style={at={(0.25,0.8)},anchor=west}
    ]
    \addplot+ table[x=cpus,y=tf,col sep=space]{Fig/dnn/scalability-4G.txt};
    \addplot+ table[x=cpus,y=oneTBB,col sep=space]{Fig/dnn/scalability-4G.txt};
    \addplot+ [mark=diamond,color=cyan] table[x=cpus,y=starpu,col sep=space]{Fig/dnn/scalability-4G.txt};
    \legend{Taskflow, oneTBB, StarPU}
    \end{axis}
  \end{tikzpicture}
  \begin{tikzpicture}[scale=0.50]
    \begin{axis}[
      title=Memory (40 CPUs),
      ylabel=MAximum RSS (GB),
      xlabel=Number of GPUs,
      legend pos=north west,
      xtick={1, 2, 3, 4},
      legend style={at={(0.50,0.55)},anchor=west},
      y filter/.code={\pgfmathparse{#1/1000000}\pgfmathresult}
    ]
    \addplot+ table[x=gpus,y=tf-m,col sep=space]{Fig/dnn/scalability-40C.txt};
    \addplot+ table[x=gpus,y=oneTBB-m,col sep=space]{Fig/dnn/scalability-40C.txt};
    \addplot+ [mark=diamond,color=cyan] table[x=gpus,y=starpu-m,col sep=space]{Fig/dnn/scalability-40C.txt};
    \legend{Taskflow, oneTBB, StarPU}
    \end{axis}
  \end{tikzpicture}
  \begin{tikzpicture}[scale=0.50]
    \begin{axis}[
      title=Memory (4 GPUs),
      ylabel=MAximum RSS (GB),
      xlabel=Number of CPUs,
      legend pos=north west,
      xtick={1, 2, 4, 8, 16, 32},
      legend style={at={(0.50,0.45)},anchor=west},
      y filter/.code={\pgfmathparse{#1/1000000}\pgfmathresult}
    ]
    \addplot+ table[x=cpus,y=tf-m,col sep=space]{Fig/dnn/scalability-4G.txt};
    \addplot+ table[x=cpus,y=oneTBB-m,col sep=space]{Fig/dnn/scalability-4G.txt};
    \addplot+ [mark=diamond,color=cyan] table[x=cpus,y=starpu-m,col sep=space]{Fig/dnn/scalability-4G.txt};
    \legend{Taskflow, oneTBB, StarPU}
    \end{axis}
  \end{tikzpicture}
  \caption{Runtime and memory data of the LSDNN (1920 layers, 4096 neurons per layer) under different CPU and GPU numbers}
  \label{fig::lsdnn_scalability}
\end{figure}


Figure \ref{fig::lsdnn_scalability} compares the performance
of solving a 1920-layered LSDNN each of 4096 neurons
under different CPU and GPU numbers.
Taskflow outperforms oneTBB and StarPU in all aspects.
Both our runtime and memory scale better regardless of the CPU and GPU numbers.
Using 4 GPUs, when performance saturates at 4 CPUs,
we do not suffer from further runtime growth as oneTBB and StarPU.
This is because our work-stealing algorithm more efficiently control wasteful steals
upon available task parallelism.
On the other hand, our memory usage is 1.5-1.7$\times$ less than oneTBB and StarPU.
This result highlights the benefit of our condition task, 
which integrates iterative control flow into a cyclic HTDG,
rather than unrolling it statically across iterations.

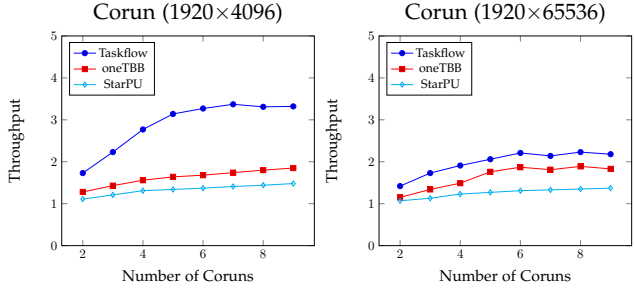
\begin{figure}[!h]
  \centering
  \pgfplotsset{
    title style={font=\LARGE},
    label style={font=\large},
  }
  \begin{tikzpicture}[scale=0.49]
    \begin{axis}[
      title= Corun (1920$\times$4096),
      ylabel=Throughput,
      xlabel=Number of Coruns,
      legend pos=north west,
      ymin=0,
      ymax=5
    ]
    \addplot+ table[x=corun,y=tf-50,col sep=space]{Fig/dnn/throughput.txt};
    \addplot+ table[x=corun,y=oneTBB-50,col sep=space]{Fig/dnn/throughput.txt};
    \addplot+ [mark=diamond,color=cyan] table[x=corun,y=starpu-50,col sep=space]{Fig/dnn/throughput.txt};
    \legend{Taskflow, oneTBB, StarPU}
    \end{axis}
  \end{tikzpicture}
  \begin{tikzpicture}[scale=0.49]
    \begin{axis}[
      title= Corun (1920$\times$65536),
      ylabel=Throughput,
      xlabel=Number of Coruns,
      legend pos=north west,
      ymin=0,
      ymax=5,
    ]
    \addplot+ table[x=corun,y=tf-100,col sep=space]{Fig/dnn/throughput.txt};
    \addplot+ table[x=corun,y=oneTBB-100,col sep=space]{Fig/dnn/throughput.txt};
    \addplot+ [mark=diamond,color=cyan] table[x=corun,y=starpu-100,col sep=space]{Fig/dnn/throughput.txt};
    \legend{Taskflow, oneTBB, StarPU} 
    \end{axis}
  \end{tikzpicture}
  \caption{Throughput of corunning inference workloads on two 1920-layered neural networks, one with 4096 neurons per layer and another with 65536 neurons per layer.}
  \label{fig::inference_throughput}
\end{figure}

We next compare the throughput of each implementation 
by corunning the same inference program
to study the inter-operability of an implementation.
We corun the same inference program up to nine processes that compete 
for 40 CPUs and 4 GPUs.
We use weighted speedup to measure the throughput.
Figure \ref{fig::inference_throughput} plots the throughput
of corunning inference programs on two different sparse neural networks.
Taskflow outperforms oneTBB and StarPU across all coruns.
oneTBB is slightly better than StarPU because StarPU 
tends to keep all workers busy all the time and results in large numbers
of wasteful steals.
The largest difference is observed 
at five coruns of inferencing the 1920$\times$4096 neural network,
where our throughput is 1.9$\times$ higher than oneTBB
and 2.1$\times$ higher than StarPU.
These CPU- and GPU-intensive workloads
highlight the effectiveness of our heterogeneous work stealing.
By keeping a per-domain invariant,
we can control cross-domain wasteful steals
to a bounded value at any time during the execution.

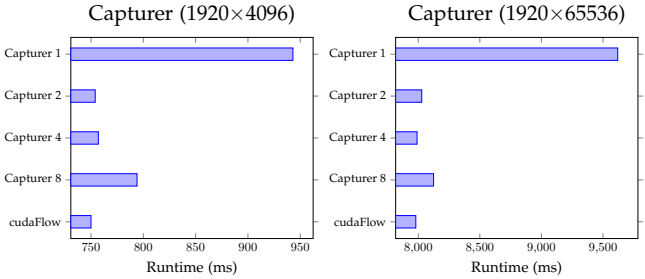
\begin{figure}[!h]
  \centering
  \pgfplotsset{
    title style={font=\LARGE},
    label style={font=\large},
  }
  \begin{tikzpicture}[scale=0.47]
    \begin{axis}[
      title= Capturer (1920$\times$4096),
      xlabel=Runtime (ms),
      legend pos=north west,
      xbar,
      symbolic y coords = {cudaFlow, Capturer 8, Capturer 4, Capturer 2, Capturer 1},
    ]
    \addplot coordinates { (750,cudaFlow) (794,Capturer 8) (757,Capturer 4)
           (754,Capturer 2) (943,Capturer 1) };
    \end{axis}
  \end{tikzpicture}
  \begin{tikzpicture}[scale=0.47]
    \begin{axis}[
      title= Capturer (1920$\times$65536),
      xlabel=Runtime (ms),
      legend pos=north west,
      xbar,
      symbolic y coords = {cudaFlow, Capturer 8, Capturer 4, Capturer 2, Capturer 1},
    ]
    \addplot coordinates { (7980,cudaFlow) (9620,Capturer 1)
           (8028,Capturer 2) (7990,Capturer 4) (8124,Capturer 8) };
    \end{axis}
  \end{tikzpicture}
  \caption{Performance of our cudaFlow capturer using 1, 2, 4, and 8 streams to complete 
  the inference of two neural networks.}
  \label{fig::optimize}
\end{figure}

We study the performance of our cudaFlow capturer 
using different numbers of streams (i.e., $max\_streams$).
For complex GPU workloads like Figure \ref{fig::snig},
stream concurrency is crucial to GPU performance.
As shown in Figure \ref{fig::optimize},
explicit construction of a CUDA graph using cudaFlow achieves
the best performance, 
because the CUDA runtime can dynamically decide the stream concurrency
with internal optimization.
For applications that must use existing stream-based APIs,
our cudaFlow capturer achieves comparable performance as cudaFlow
by using two or four streams.
Taking the 1920$\times$65536 neural network for example,
the difference between our capturer of four streams and cudaFlow
is only 10 ms.
For this particular workload,
we do not observe any performance benefit beyond four streams.
Application developers can fine-tune this number.

\begin{figure}[!h]
  \centering
  \pgfplotsset{
    title style={font=\LARGE},
    label style={font=\large},
  }
  \begin{tikzpicture}[scale=0.49]
    \begin{axis}[
      title= cudaFlow vs Stream (runtime),
      ylabel=Runtime,
      ymax=4000,
      xlabel=GPU Count,
      xtick={1, 2, 3, 4},
      legend pos=north east,
    ]
    \addplot+ table[x=gpu,y=tf-cpu,col sep=space]{Fig/dnn/cudaflow_overhead.txt};
    \addplot+ [mark=o,color=blue] table[x=gpu,y=s1-cpu,col sep=space]{Fig/dnn/cudaflow_overhead.txt};
    \addplot+ [mark=x,color=blue] table[x=gpu,y=s2-cpu,col sep=space]{Fig/dnn/cudaflow_overhead.txt};
    \addplot+ [mark=+,color=blue] table[x=gpu,y=s4-cpu,col sep=space]{Fig/dnn/cudaflow_overhead.txt};
    \legend{cudaFlow, CUDA Stream 1, CUDA Stream 2, CUDA Stream 4}
    \end{axis}
  \end{tikzpicture}
  \begin{tikzpicture}[scale=0.49]
    \begin{axis}[
      title= cudaFlow vs Stream (mem),
      xlabel=GPU Count,
      ylabel=Maximum RSS (GB),
      legend pos=south east,
      xtick={1, 2, 3, 4},
      y filter/.code={\pgfmathparse{#1/100000}\pgfmathresult}
    ]
    \addplot+ table[x=gpu,y=tf-mem,col sep=space]{Fig/dnn/cudaflow_overhead.txt};
    \addplot+ [mark=o,color=blue]table[x=gpu,y=s1-mem,col sep=space]{Fig/dnn/cudaflow_overhead.txt};
    \addplot+ [mark=x,color=blue] table[x=gpu,y=s2-mem,col sep=space]{Fig/dnn/cudaflow_overhead.txt};
    \addplot+ [mark=+,color=blue] table[x=gpu,y=s4-mem,col sep=space]{Fig/dnn/cudaflow_overhead.txt};
    \legend{cudaFlow, CUDA Stream 1, CUDA Stream 2, CUDA Stream 4}
    \end{axis}
  \end{tikzpicture}
  \caption{Comparison of runtime and memory between cudaFlow (CUDA Graph) and stream-based execution.}
  \label{fig::dnn_cudaflow_overhead}
\end{figure}
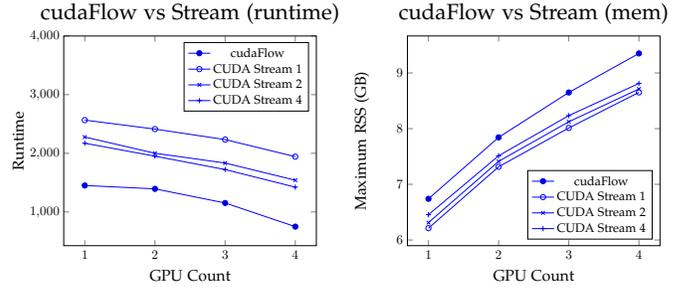

We finally compare the performance of cudaFlow with stream-based execution.
As shown in Figure \ref{fig::dnn_cudaflow_overhead},
the line cudaFlow represents our default implementation using explicit CUDA graph
construction, and
the other lines represent stream-based implementations for the same task graph
using one, two, and four streams.
The advantage of CUDA Graph is clearly demonstrated in this large machine learning
workload of over 2K dependent GPU operations per cudaFlow.
Under four streams that deliver the best performance for the baseline,
cudaFlow is 1.5$\times$ (1451 vs 2172) faster at one GPU and is 1.9$\times$ (750 vs 1423) 
faster at four GPUs.
The cost of this performance improvement is increased memory usage
because CUDA Graph needs to store all the operating parameters
in the graph.
For instance, under four streams,
cudaFlow has 4\% and 6\% higher memory usage than stream-based execution at
one and four GPUs, respectively.


\section{Related Work}


\subsection{Heterogeneous Programming Systems}

Heterogeneous programming systems
are the main driving force to advance scientific computing.
Directive-based programming models~\cite{OmpSs, OpenACC, OpenMP, XKAAPI, OpenMPC}
allow users to augment program information 
of loop mapping onto CPUs/GPUs and data sharing rules
to designated compilers for automatic parallel code generation.
These models are good at loop-based parallelism
but cannot handle irregular task graph patterns efficiently~\cite{Lee_12_01}.
Functional approaches~\cite{TBB, Fastflow, StarPU, Legion, PaRSEC, Kokkos, Huang_21_02, HPX, Charm++, Lima_15_01}
offer either implicit or explicit task graph constructs
that are more flexible in runtime control and on-demand tasking.
Each of these systems has its pros and cons.
However, few of them enable end-to-end expressions of heterogeneously dependent
tasks with general control flow.

\subsection{Heterogeneous Scheduling Algorithms} 

Among various heterogeneous runtimes,
work stealing is a popular strategy to reduce the complexity of 
load balancing~\cite{ABP, Lima_15_01} 
and has inspired the designs of many 
parallel runtimes~\cite{TBB, Nabbit, TPL, Cilk++, X10}.
A key challenge in work-stealing designs is worker management.
Instead of keeping all workers busy most of the time~\cite{ABP, StarPU, Huang_21_02},
both oneTBB~\cite{TBB} and BWS~\cite{BWS} have developed sleep-based strategies.
oneTBB employs a mixed strategy of fixed-number worker notification,
exponential backoff, and noop assembly.
BWS modifies OS kernel to alter the yield behavior.
\cite{Lin_20_01} takes inspiration from BWS and oneTBB to develop an 
adaptive work-stealing algorithm to minimize the number of wasteful steals.
Other approaches, such as \cite{A-STEAL} that targets a space-sharing environment,
\cite{EWS} that tunes hardware frequency scaling,
\cite{SadayWS, LifelineLoadBalancing} that balance load on distributed memory,
\cite{LAWS, Yi_10_01, Suksompong_16_01, Han_18_01} that deal with data locality,
and \cite{DeterministicWS} that focuses on memory-bound applications
have improved work stealing in certain performance aspects, but their results
are limited to the CPU domain.
How to migrate the above approaches to a heterogeneous target remains an open question.

In terms of GPU-based task schedulers,
Whippletree~\cite{Whippletree} design a fine-grained resource scheduling algorithm
for sparse and scattered parallelism atop a custom program model.
\cite{RLGPUPlacement} leverages reinforcement learning to place machine learning workloads
onto GPUs.
Hipacc~\cite{Hipacc} introduces a pipeline-based optimization for CUDA graphs
to speed up image processing workloads.
\cite{Yu_20_01} develops a compiler to transforms OpenMP directives to a CUDA graph.
These works have primarily focused on scheduling GPU tasks in various applications,
which are orthogonal to our generic heterogeneous scheduling approaches.


\section{Acknowledgements}

The project is supported by the DARPA contract FA 8650-18-2-7843
and the NSF grant CCF-2126672.
We appreciate all Taskflow contributors and 
reviewers' comments for improving this paper.

\section{Conclusion}

In this paper, we have introduced Taskflow, 
a lightweight task graph computing system
to streamline the creation of heterogeneous programs with control flow.
Taskflow has introduced a new programming model that enables
an end-to-end expression of heterogeneously dependent tasks 
with general control flow.
We have developed an efficient work-stealing runtime optimized for
latency, energy efficiency, and throughput,
and derived theory results to justify its efficiency.
We have evaluated the performance of Taskflow on 
both micro-benchmarks and real applications.
As an example,
Taskflow solved a large-scale machine learning problem up to 
29\% faster,
1.5$\times$ less memory,
and 1.9$\times$ higher throughput
than the industrial system, oneTBB,
on a machine of 40 CPUs and 4 GPUs.

Taskflow is an on-going project under active development.
We are currently exploring three directions:
First, we are designing a distributed tasking model based on partitioned
taskflow containers with each container running on a remote machine.
Second, we are extending our model to incorporate SYCL~\cite{SYCL}
to provide a single-source heterogeneous task graph programming environment.
The author Dr. Huang is a member of SYCL Advisory Panel and is collaborating
with the working group to design a new SYCL Graph abstraction.
Third, we are researching automatic translation methods
between different task graph programming models using
Taskflow as an intermediate representation.
Programming-model translation has emerged as an important research area in 
today's diverse computing environments because no one programming model
is optimal across all applications.
The recent 2021 DOE X-Stack program directly calls for novel
translation methods to facilitate performance optimizations 
on different computing environments.

One important future direction is to collaborate with Nvidia CUDA teams
to design a conditional tasking interface within the CUDA Graph itself.
This design will enable efficient control-flow decisions to be made 
completely in CUDA runtime,
thereby largely reducing the control-flow cost
between CPU and GPU.


\bibliographystyle{plain}
\bibliography{ms}

\begin{IEEEbiography}[{\includegraphics[width=1in,height=1.25in,clip,keepaspectratio]{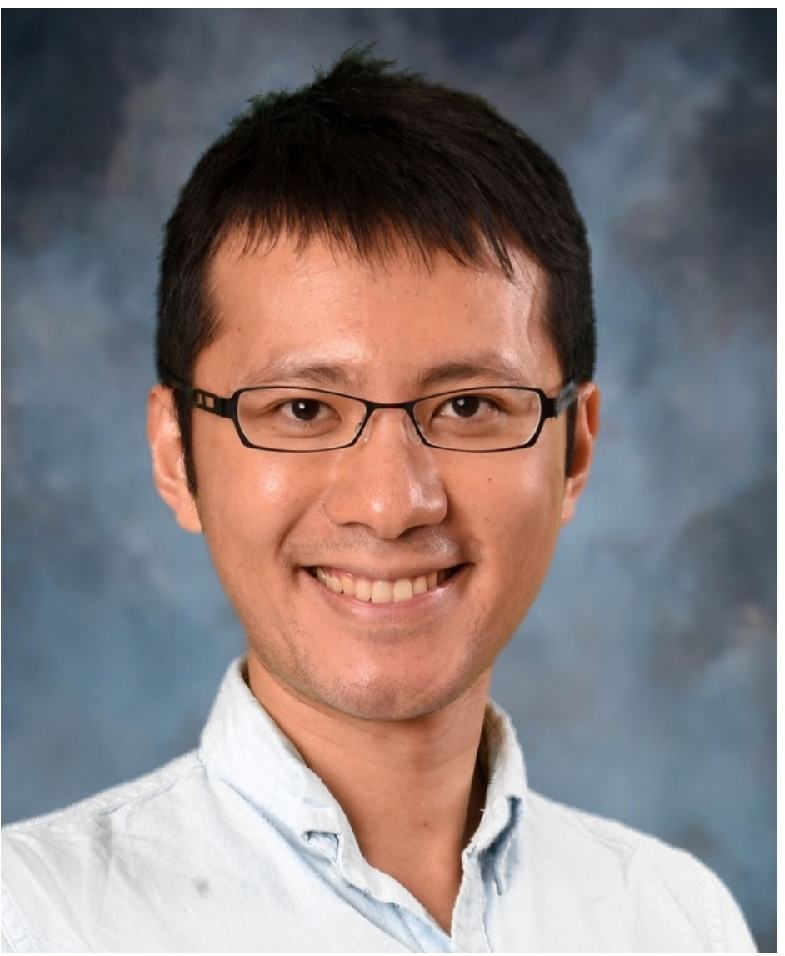}}]{Tsung-Wei Huang}
received the B.S. and M.S. degrees from the Department of Computer Science, 
National Cheng Kung University (NCKU), Tainan, Taiwan, in 2010 and 2011, respectively.
He obtained his Ph.D. degree in the
Electrical and Computer Engineering (ECE) Department
at the University of Illinois at Urbana-Champaign (UIUC). 
He is currently an assistant professor in the ECE department at the
University of Utah.
Dr. Huang has been building software systems for parallel computing and timing analysis.
His PhD thesis won the prestigious 2019 ACM SIGDA Outstanding PhD Dissertation Award
for his contributions to distributed and parallel VLSI timing analysis.
\end{IEEEbiography}

\begin{IEEEbiography}[{\includegraphics[width=1in,height=1.25in,clip,keepaspectratio]{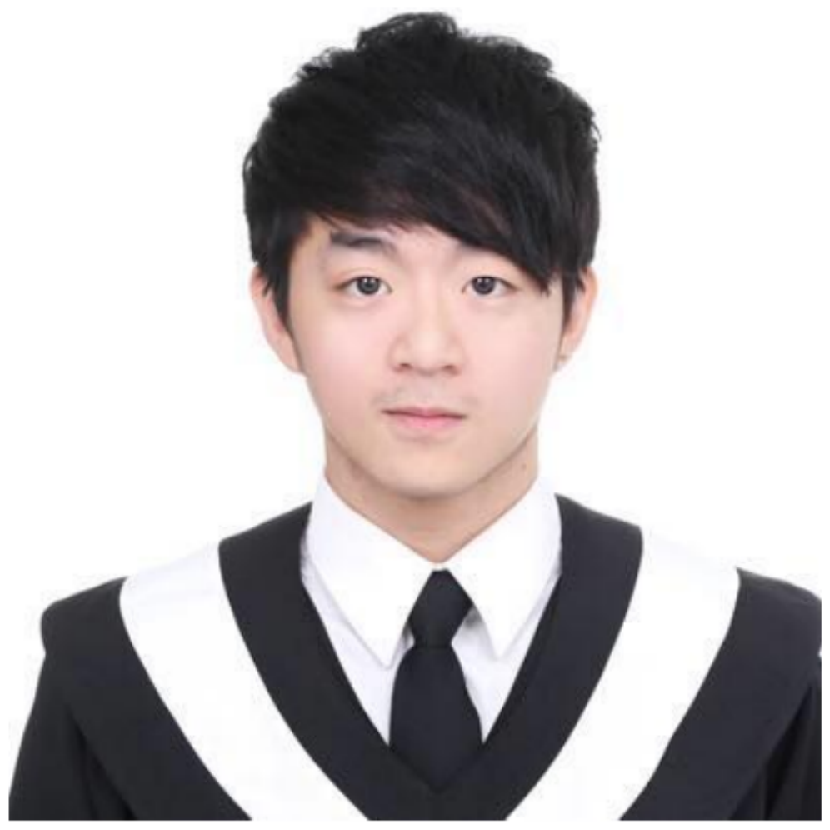}}]{Dian-Lun Lin}
received the B.S.~degree from the Department of Electrical Engineering at Taiwan's Cheng Kung University and M.S.~degree from the Department of Computer Science at National Taiwan University.
He is current a Ph.D. student at the Department of Electrical and Computer Engineering at the University of Utah.
His research interests are in parallel and heterogeneous computing with a specific focus on CAD applications.
\end{IEEEbiography}

\begin{IEEEbiography}[{\includegraphics[width=1in,height=1.25in,clip,keepaspectratio]{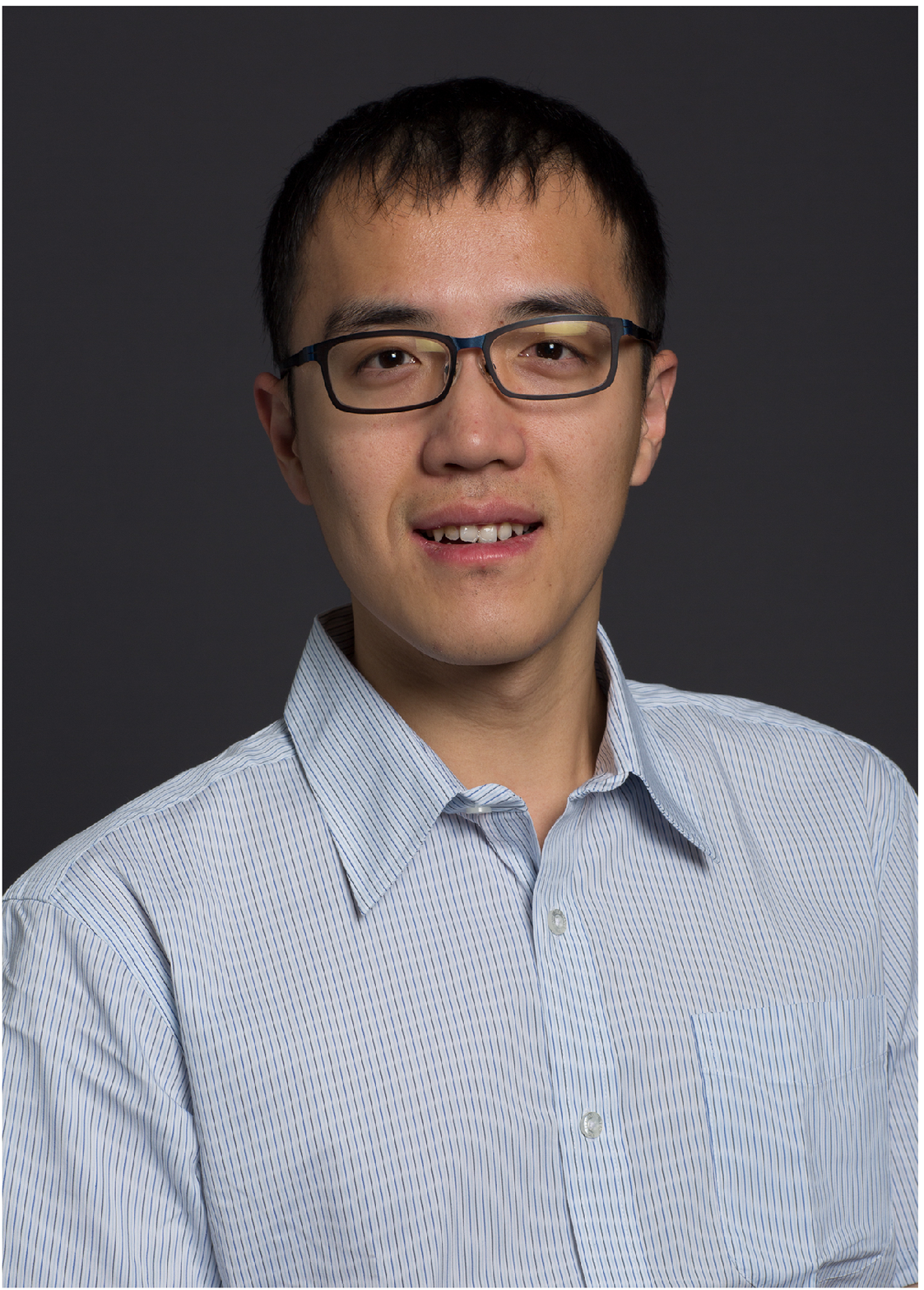}}]{Chun-Xun Lin}
received the B.S. degree in Electrical Engineering from the National Cheng Kung University,
Tainan, Taiwan, and the M.S. degree in Electronics Engineering from the Graduate Institute of Electronics
Engineering, National Taiwan University, Taipei, Taiwan, in 2009 and 2011, respectively. 
He received his Ph.D. degree from the department of
Electrical and Computer Engineering (ECE) at the University of Illinois
at Urbana-Champaign (UIUC) in 2020.
His research interest is in parallel processing.
\end{IEEEbiography}

\begin{IEEEbiography}[{\includegraphics[width=1in,height=1.25in,clip,keepaspectratio]{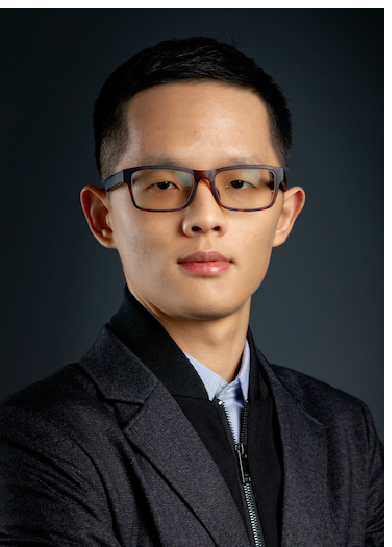}}]{Yibo Lin}
(S’16--M’19)
received the B.S.~degree in microelectronics from Shanghai Jiaotong University in 2013,
and his Ph.D. degree from the Electrical and Computer Engineering Department of the University of Texas at Austin in 2018.
He is current an assistant professor in the Computer Science Department associated with the Center for Energy-Efficient Computing and Applications at
Peking University, China.
His research interests include physical design, machine learning applications, GPU acceleration, and hardware security.
\end{IEEEbiography}

%
%
%
%

\end{document}